\documentclass[11pt,reqno,notitlepage]{article}

\usepackage{amssymb}
\usepackage{amsthm}
\usepackage{graphicx}
\usepackage{amsmath}
\usepackage{fancyhdr}
\usepackage{natbib}
\usepackage{datetime}
\usepackage[toc]{appendix}
\usepackage{tikz}
\usepackage{mathrsfs}
\usepackage[normalem]{ulem}
\usepackage{graphicx, xcolor}
\usepackage[breaklinks]{hyperref}

\usepackage{setspace}
\usepackage{graphicx, xcolor}												
\usepackage[mathscr]{eucal}												
\usepackage{subcaption}
\usepackage{epstopdf}

\usetikzlibrary{arrows,decorations.pathreplacing}
\usepackage[top=1in, bottom=1in, left=1in, right=1in]{geometry}
\usepackage{mathtools}									%
\mathtoolsset{showonlyrefs=true}							%

\newtheorem{theorem}{Theorem}
\newtheorem{proposition}[theorem]{Proposition}	
	
\newtheorem{assumption}[theorem]{Assumption}	
\newtheorem{definition}[theorem]{Definition}
\newtheorem{remark}{Remark}
\numberwithin{equation}{section}	
\numberwithin{theorem}{section}
\newcommand{\blue}[1]{{\color{blue}{{#1}}}}

\newcommand{\bp}{\begin{pmatrix} } 
\newcommand{\ep}{\end{pmatrix}} 

\renewcommand{\(}{\left(}				
\renewcommand{\)}{\right)}
\renewcommand{\[}{\left[}
\renewcommand{\]}{\right]}	
\renewcommand{\vec}[1]{\mathbf{#1}}	
\def \norm#1{\left\| #1 \right\| }
\def\q{\vec{q}}
\def\s{\vec{s}}
\def\f{\vec{f}}
\def\h{\vec{h}}
\def\a{\vec{a}}
\def\Phiv{\vec{\Phi}}

\def\R{\mathbb{R}}
											
\def\I{\mathbb{I}}

\def\H{\mathcal{H}}

\def\D{\mathcal{D}}	
\def\Q{\mathcal{Q}}

\def\Qh{\widehat{\Q}}

\def\d{\partial}		

\def\ind{\mathbb{I}}

\def \abs#1{\left| #1 \right| }

\newcommand{\diag}{\operatorname{diag}}

\def\define{:=}

\def\T{\top}

\DeclareMathOperator*{\argmin}{arg\,min}

\newtheorem{example}[theorem]{Example}

\DeclareMathOperator*{\FIX}{FIX}

\begin{document}

\title{A Repo Model of Fire Sales with VWAP and LOB Pricing Mechanisms}

\author{
Maxim Bichuch
\thanks{
Department of Applied Mathematics and Statistics,
Johns Hopkins University
3400 North Charles Street, 
Baltimore, MD 21218. 
{\tt mbichuch@jhu.edu}. Work  is partially supported by NSF grant DMS-1736414. Research is partially supported by the Acheson J. Duncan Fund for the Advancement of Research in Statistics.}
\and  Zachary Feinstein
\thanks{
School of Business
Stevens Institute of Technology, Hoboken
NJ 07030, USA,
{\tt  zfeinste@stevens.edu}. \emph{Corresponding author}. }
}
\date{\today}
\maketitle

\begin{abstract}
We consider a network of banks that optimally choose a strategy of asset liquidations and borrowing in order to cover short term obligations. The borrowing is done in the form of collateralized repurchase agreements, the haircut level of which depends on the total liquidations of all the banks. Similarly the fire-sale price of the asset obtained by each of the banks depends on the amount of assets liquidated by the bank itself and by other banks. By nature of this setup, banks' behavior is considered as a Nash equilibrium. This paper provides two forms for market clearing to occur: through a common closing price and through an application of the limit order book.  The main results of this work are providing the existence of maximal and minimal clearing solutions (i.e., liquidations, borrowing, fire sale prices, and haircut levels) as well as sufficient conditions for uniqueness of the clearing solutions. 
\end{abstract}

\section{Introduction}
Historically, financial risk was typically measured for individual firms separately. After the financial crisis of 2007-2009, a new understanding that risk can spread through the entire financial system has emerged. This is referred to as systemic risk -- the risk that the distress of several banks can spread throughout the system to a degree that it may affect the viability of the entire system or a significant part of it. Such a propagation of risk is known as financial contagion. Two types of contagion are usually distinguished: those that happen due to local connections (e.g., obligations between banks in the network), and those that happen due to their influence on the entire network (e.g., impact to asset prices). This study focuses on a form of global contagion through asset prices, and investigates the existence maximal and minimal Nash equilibria in a model of fire sales and collateralized borrowing. We further provide sufficient conditions for uniqueness of the Nash equilibrium. 


Price mediated contagion, i.e., systemic risk spreads through the market by impacting  the asset liquidations (and purchases).  
As prices drop due to a fire sale by one market participant to meet an obligation or satisfy some regulation, the value of the assets held by all other institutions are also impacted due to, e.g., mark-to-market accounting rules. 
Due to these writedowns in the value of assets, these other institutions in turn may themselves need to sell assets in a fire sale to meet their own obligations and satisfy regulatory constraints.  This is a type of global contagion, as these writedowns impact all banks that hold the asset, and therefore it is fundamental to systemic risk.  Price mediated contagion through regulatory constraints such as leverage requirements were studied in equilibrium models by, e.g.,~\cite{CFS05,GLT15,CL15,feinstein2015leverage,braouezec2016risk,braouezec2017strategic,DE18,feinstein2019leverage,CS19,banerjee2019price}.  In this work we focus instead on fire sales which are precipitated by the need to cover short-term obligations; this problem was studied in equilibrium models by, e.g.,~\cite{caccioli2014,AFM16,AW_15,feinstein2015illiquid,feinstein2017multilayer,CLY14}.  We especially wish to highlight the works of, e.g.,~\cite{caballero2013fire,feinstein2015illiquid,feinstein2017multilayer,braouezec2017strategic,FH2020} which model the fire sales as a Nash equilibrium.

In this paper we extend the model of \cite{bichuch2019optimization}; that work considered a network of banks facing shortfalls on their obligations which can be met through borrowing or by liquidating assets in which each firm had an infinite capacity to borrow.  (We wish to note that none of the references provided in the prior paragraph allow for firms to borrow to cover their obligations.) 
The primary goal of this current paper is to consider the effects of repurchase agreement (repo) markets on financial stability.  Such markets require banks to post collateral above the value of the loan in order to secure short term financing; this is described by a \emph{haircut} on the value of the asset when used as collateral.  In this construction, each bank seeks to optimize their strategy between asset liquidations and borrowing in the repo market.  As in the traditional fire sale literature (see, e.g., \cite{AFM16}), asset liquidations cause price impacts and, thus, the actions of one bank influence the decisions of all other institutions as well, i.e., we consider the Nash equilibrium of strategies.  
Based on the static setting traditionally followed in the literature (see, e.g.,~\cite{CFS05,GLT15,AFM16}), we assume that all the trading happens simultaneously and instantaneously.  

It is often the case that there is no unique Nash equilibrium, or alternatively the conditions for uniqueness turn out to be much stronger than the ones required for existence. An important intermediate step is to establish a monotonicity result, i.e. maximal and minimal solutions. This is done in many classical works, e.g. \cite{EN01,RV13,braouezec2019strategic}. In this paper, in addition to providing a set of sufficient conditions for uniqueness of Nash equilibrium, we also show find the maximal and minimal Nash equilibria under the same conditions as the one we assume for its existence.

Previously, prices were provided by an inverse demand function which was used to price liquidations as well as provide mark-to-market accounting. We refer to, e.g.,~\cite{CFS05,GLT15,AFM16} as well as in the prior modeling work of~\cite{bichuch,bichuch2019optimization} which we extend.
In undertaking this study, we consider two classical and realistic pricing functions in the fire sale process: \emph{Volume Weighted Average Pricing} (VWAP) and a \emph{Limit Order Book} (LOB) based pricing scheme.  
Both of these schemes can be viewed as pricing limits as order sizes decrease to zero, but with different rates of liquidation.  This allows us to incorporate notions of time dynamics into the static model proposed.
The VWAP scheme determines prices if firms place orders at a rate proportional to their total desired liquidations; this, ultimately, results in the same average price for every bank.  Such a pricing scheme was introduced in~\cite{banerjee2019price}.
The LOB setting distinguishes prices by assuming all firms place orders at the same speed; banks with smaller order volumes will receive a higher price than those with a larger order volume (as the latter will continue to eat through the book even after the former are done liquidating).

As highlighted above, the innovation of this work is two-fold.  First, we consider realistic pricing schemes that allow for banks to receive different prices based on the quantity of assets sold instead of the, more standard, assumption that there is a unique price at which all transactions occur.  Second, we consider collateralized borrowing of illiquid assets in a repo market in which the haircut of this collateral also depends on the mark-to-market value of the asset.  As opposed to the realized liquidation prices, the haircut remains bank independent and only depends on the entire sale volume of the entire banking system since the deal depends on the value of the collateralized asset rather than the riskiness of the individual banks.  As in reality, because the repo loan is collateralized, it is considered (practically) risk-free. Therefore, just like the loan, the collateral value also needs to be bank-independent. 
Under these constructions, we are able to investigate the sensitivity of the resulting market prices to the prevailing repo interest rate.  In particular, regulators use interest rates as the primary control for financial stability.  This was seen in the emergency liquidity injection by the Federal Reserve in September 2019, in order to stabilize the repo market \cite{ihrig2020fed,afonso2020monetary}.  In fact, \cite{GM09,B09} consider the 2007-2009 financial crisis as a run on the repo market.  Therefore systematic consideration of repo markets and the impact of interest rates is of paramount importance.

The organization of this paper is as follows. 
Section \ref{sec:model} introduces the general model with general inverse demand pricing functions.  In that section we provide the existence of Nash equilibrium under a minimal set of assumptions.  Section \ref{sec:main} introduces the VWAP and LOB based inverse demand pricing functions and discusses the conditions needed for maximal clearing solutions and uniqueness of Nash equilibrium.  
Numerical case studies and comparison of VWAP and LOB inverse demand pricing functions is in Section \ref{sec:cs}. The proofs for the main results are provided in the Appendix.
Additionally in the Appendix, under the uniqueness conditions, we investigate the sensitivity of the clearing solutions to the prevailing repo rate.

\section{Financial setting}\label{sec:model}
We begin by assuming a system of $n$ banks. In contrast to works that explicitly depend on the network of interbank obligations, e.g., in \cite{EN01,CFS05,AFM16,feinstein2015illiquid}, herein we will consider only fire sale effects and price mediated contagion as in, e.g., \cite{GLT15,braouezec2016risk,braouezec2017strategic,feinstein2019leverage,banerjee2019price}.  We will, for simplicity, assume that all the banks $i=1, ..., n$ are facing a (cash) shortfall $h_i>0$, all while holding $a_i > 0$ shares of illiquid assets; any banks without either a shortfall or illiquid asset holdings will not participate in any fire sale or borrowing and thus are extraneous to the considerations of this model. The banks are faced with the task of finding the optimal strategy to raise $h_i$ cash in order to cover this shortfall. We assume that they can do so by either selling their illiquid asset, borrowing, or both. It will be assumed that the borrowing is going to be collateralized using the same illiquid asset. 
As is standard in the literature, due to the illiquidity, the price of the illiquid asset declines as assets are being sold; this is due to supply-demand dynamics so that the equilibrium is maintained.  The same effect is assumed for the collateral value of the asset. 
To simplify the setting we only seek to model those institutions with shortfall $h_i > 0$.  That is, we are only modeling borrowers in this work; we refer to Remark~\ref{rem:r} and Appendix~\ref{sec:r} for brief discussion of how lenders may impact the model herein.

Herein we introduce two ``pricing'' functions.  Let $\bar f_i: \R_+^n \to [0,1]$ denote the average price obtained by bank $i=1, ..., n$ given the set of system liquidations $(s_1, ..., s_n) \in \D\define\prod_{j=1}^n [0,a_j]$.  Note that we implicitly impose a no short selling constraint throughout this work. Here, without loss of generality, it was assumed that the current, highest price of the asset, before any sales happened is $1$, and it can only decrease thereafter. Notably, the construction of $\bar f_i$ implies that different banks may obtain different prices in the market due to the market design or different order sizes.  Let $g: \R_+ \to [0,1]$ denote the price of the collateralized asset in the repurchase agreements under study, i.e., the function $g\(\sum_{j=1}^n s_j\)$ encodes the haircut on the asset as a mapping of the total liquidations by all the banks. Note that while the price obtained by bank $i$ may be unique due to the different quantities different banks are selling, since the repo transaction is collateralized it is assumed that the repo market offers the same repo rate $r$ to all banks and uses the same haircut $g\(\sum_{j=1}^n s_j\)$. Though we call $g$ the ``haircut'', it is more appropriate to denote $1-g$ to be the true haircut on the asset in the repo market.  At various times in this work we will refer to $g$ as the haircut and others $1-g$ will be given that name. 

Assuming banks sell $\s\define(s_1, ..., s_n) \in \D$, the realized loss to bank $i$ from the sale is $s_i(1-\bar f_i(\s))$. The bank obtained $s_i\bar f_i(\s)$ through this sale, therefore it needs to borrow an additional $(h_i -  s_i \bar f_i(\s))$ for the cost of $r(h_i -  s_i \bar f_i(\s))$. We will abuse notation and denote for convenience $\bar f_i$ to be both $\bar f_i(s_i, \s_{-i})$ and $\bar f_i(\s)$ where $\s_{-i} := (s_1,...,s_{i-1},s_{i+1},...,s_n) \in \prod_{j = 1, j \neq i}^n [0,a_j]$. Therefore, bank $i$ seeks to optimize:
\begin{align}
\label{eq:s*}\begin{split}
s_i^* = s_i^{*}(\s_{-i})&= \argmin_{s_i \in [0,a_i]} \; s_i \left(1 - \bar f_i(s_i,\s_{-i})\right) + r\left(h_i - s_i \bar f_i(s_i,\s_{-i})\right)\\
&\qquad \text{s.t.} \quad s_i \leq \frac{h_i}{\bar f_i(s_i,\s_{-i})}, \; s_i \geq \frac{h_i - a_i g\(\sum_{j=1}^n s_j\)}{\bar f_i(s_i,\s_{-i}) - g\(\sum_{j=1}^n s_j\)}.
\end{split}
\end{align}
Here, the first inequality ensures that bank $i$ does not obtain more than $h_i$ through the asset sale, and the second inequality constraint is used to ensure that $h_i - s_i \bar f_i(s_i,\s_{-i}) \leq (a_i - s_i) g\(\sum_{j = 1}^n s_j\)$, i.e., after the sale, bank $i$ has enough collateral $(a_i-s_i)g\(\sum_{j=1}^n s_j\)$ to cover its loan. 
The paper of \cite{bichuch2019optimization} considers the case in which no haircut is taken, i.e., $g\equiv 1$.  

In~\eqref{eq:s*}, it follows that bank $i$ is solvent if and only if 
\begin{equation*}
\frac{h_i - a_i g\(\sum_{j = 1}^n s_j\)}{\bar f_i(s_i,\s_{-i}) - g\(\sum_{j = 1}^n s_j\)} \leq \frac{h_i}{\bar f_i(s_i,\s_{-i})} \leq a_i.
\end{equation*}
By construction of the haircut for repurchase agreements $0 \leq g\(\sum_{j = 1}^n s_j\) < \bar f_i(s_i,\s_{-i})$.  Under such a construction bank $i$ is solvent if and only if $h_i \leq a_i \bar f_i(s_i,\s_{-i})$, i.e., if at the current price realized by bank $i$ it is possible for said bank to cover its shortfall by liquidations alone.  If bank $i$ is insolvent then we will assume that it is forced to liquidate all of its asset holdings, i.e., $s_i^* = a_i$.

\begin{remark}
With the construction of~\eqref{eq:s*}, we wish to highlight the innovations of this model.  
Fire sales with borrowing was introduced in~\cite{bichuch2019optimization}; however, that paper assumed that borrowing could occur costlessly, i.e., without any haircut.  The introduction of the haircut function and the associated constraint $s_i \geq \frac{h_i - a_i g\(\sum_{j=1}^n s_j\)}{\bar f_i(s_i,\s_{-i}) - g\(\sum_{j=1}^n s_j\)}$ encodes the key notion of collateralized borrowing in a repo market.  In fact, this haircut constraint puts a \emph{lower bound} on the liquidations of the banks in the system (dependent on the total amount of assets being sold) which was not considered previously in~\cite{bichuch2019optimization}.  This new constraint, additionally, allows us to have an endogenous definition for bank solvency ($h_i \leq a_i \bar f_i(s_i,\s_{-i})$); this is in contrast with~\cite{bichuch2019optimization} which, a priori, assumed all banks under consideration that needed to borrow could do so with a haircut of $g \equiv 1$.  Not only does this haircut function more accurately model collateralized borrowing in a repo market, it also allows for a notion of loans dependent on the ``quality'' of the asset as measured by the total amount being liquidated (or, alternatively, a function of the last price quoted in the market).
Additionally, as compared to the traditional fire sale literature (e.g.,~\cite{CFS05,GLT15,AFM16}), we consider the problem in which each bank may have a different price as determined by the collection of inverse demand functions $\bar f_i$.  This can be due to, e.g., the use of limit order book for all liquidations as occurs in reality; such a formulation for that specific setting is presented in the next section.  
\end{remark}
For convenience, for the remainder of this work, denote $\bar q_i = \bar f_i(s_i,\s_{-i}),~i=1, ...,n,$ and $q = g\(\sum_{j = 1}^n s_j\)$.  With this notation, we modify~\eqref{eq:s*} (similarly as in~\cite{bichuch2019optimization}) such that we seek a Nash equilibrium of the game for each bank $i$
\begin{align}
\label{eq:s*q}
s_i^* =s_i^* (\s_{-i}, q, \bar \q)&= \argmin_{s_i \in [0,a_i]} \; s_i \left(1 - \bar f_i(s_i,\s_{-i})\right) + r\left(h_i - s_i \bar f_i(s_i,\s_{-i})\right) \\
&\qquad \text{s.t.} \quad s_i \leq \frac{h_i}{\bar q_i}, \; s_i \geq \frac{h_i - a_i q}{\bar q_i - q}.
\end{align}
The goal is then to find a Nash equilibrium for \eqref{eq:s*q}, such that $(q,\bar \q)$ are, additionally, fixed points of 
\begin{align}
q &= g\(\sum_{j = 1}^n s_j^*\), \qquad \bar q_i = \bar f_i(\s^*).
\label{eq:q-fixed}
\end{align}

As noted above, bank $i$ is defaulting if $h_i \geq a_i \bar q_i$ and, in such a situation, $s_i^* = a_i$. Our goal is primarily to find conditions for existence and uniqueness of this Nash game in the financial system. In order to do that we need assumptions on the inverse demand functions $\bar f_i$ and $g$.

\begin{assumption}\label{ass:bar-f-g}
Let $M \geq \sum_{i = 1}^n a_i$ be the total initial market capitalization of the illiquid asset. For $i=1, ..., n$ we assume that 
\begin{enumerate}
\item $\bar f_i\colon \D \to (0,1]$ are each continuous and non-increasing in every argument with $\bar f_i(0, ..., 0) = 1$. 
\item For $\s_{-i} \in \prod_{j=1,j\neq i}^n [0,a_j]$ we assume that $s_i \in[0,a_i]\mapsto s_i \bar f_i(s_i,\s_{-i})$ is strictly increasing and concave.
\item\label{ass:bar-f-g-3} The haircut function $g: [0,M] \to (0,1]$ is continuous, convex, and strictly decreasing, with $\min_{1\le i \le n}\bar f_i(\s) > g\(\sum_{j = 1}^n s_j\)$ for every $\s \in \D$.
\item\label{ass:bar-f-g-4} For $\s_{-i} \in \prod_{j=1,j\neq i}^n [0,a_j]$, the mapping $s_i \in[0,a_i]\mapsto s_i \bar f_i(s_i,\s_{-i}) + (a_i-s_i)g\(\sum_{j = 1}^n s_j\)$ is strictly increasing.
\end{enumerate}
\end{assumption}

The intuitive meaning of the first assumption is that the prices are declining as sales increase. The intuition behind the second assumption is that the greater quantity being sold, the more cash can be obtained. The third assumption is similar to the first, and additionally, we assume that the haircut, is strictly less than the smallest price of the asset. Finally, the intuition behind the last assumption is that the total cash that can be raised and borrowed increases with the amount being sold.

We are interested in investigating the equilibrium of problem \eqref{eq:s*}, however because of the dependency of the constraints on the actions of the other banks, this is not so easy. Therefore we investigate the equilibria of the problem \eqref{eq:s*q} and \eqref{eq:q-fixed} instead. While the two problems are not identical, the connection between them, under the Assumptions \ref{ass:bar-f-g}, is given in the following proposition, the proof of which is delayed until Appendix~\ref{sec:proof-sec2}.
\begin{proposition}\label{prop:nash-equiv}
Under Assumption \ref{ass:bar-f-g} a Nash equilibrium  $\s^{**} \in\D$ of \eqref{eq:s*q} with equilibrium prices $(q^{**},\bar q^{**}_1, ..., \bar q^{**}_n) = \(g\(\sum_{i = 1}^n s_i^{**}\) , \bar f_1\(\s^{**}\), ...,  \bar f_n\(\s^{**}\)\)$ is also a Nash equilibrium of problem \eqref{eq:s*}, and vice versa.
\end{proposition}

Existence of a Nash equilibrium easily follows as a consequence of Brouwer's fixed-point theorem:
\begin{theorem}[Existence of Nash Equilibrium]\label{thm:exist}
Assume the inverse demand functions $\bar f_i,~i=1,...,n$ and haircut function $g$ satisfy Assumption \ref{ass:bar-f-g}. Then there exists a Nash equilibrium liquidating strategy $\s^{**} \in\D$ with equilibrium prices $(q^{**},\bar q^{**}_1, ..., \bar q^{**}_n) = \(g\(\sum_{i = 1}^n s_i^{**}\) , \bar f_1\(\s^{**}\), ...,  \bar f_n\(\s^{**}\)\)$.
\end{theorem}
\begin{proof}[Proof of Theorem \ref{thm:exist}]
Fix bank $i$ and consider \eqref{eq:s*q} as a function of $(\s_{-i},q,\bar q_i)$ such that $0 \leq q < \bar q_i$, with $\bar f_i(a_1,...,a_n) \leq \bar q_i$, and $\s_{-i}^* \in \prod_{j = 1, \neq i}^n [0,a_j]$.  Since the objective function of \eqref{eq:s*q} is convex in $s_i$ and the constraint set is a convex interval, the set of minimizers for a fixed set of parameters $(\s_{-i},q,\bar q_i)$ is convex.  
An application of Berge maximum theorem (on $\bar q_i \geq \frac{h_i}{a_i}$ due to the continuity of the objective and constraint functions) yields upper continuity and convex-valuedness of the set of minimizers.  This is extended for the region of insolvency by the assumption that $s_i^* = a_i$ on $h_i > a_i \bar q_i$.  Thus a joint equilibrium $(\s^{**},q^{**},\bar q_1^{**},...,\bar q_n^{**})$ can be found via Kakutani's fixed point theorem.
\end{proof}

It turns out that the conditions for existence of equilibrium are very mild, compared to the uniqueness conditions.  This is not surprising considering the following example.
\begin{example}\label{ex:counterexample}
Consider an $n = 2$ bank setting with $r = 0$ repo rate.  Let both banks have the same capitalization $a$ and shortfall $h$.  Let $\bar f_1(\s) = \bar f_2(\s) = \hat f(s_1 + s_2)$ for any $\s \in \D$ and such that $a \hat f(2a) < h < a g(0)$ (e.g., $\hat f(s) = 1 - \frac{s}{4a}$ and $g(s) = 0.7 - \frac{s}{4a}$ with $h \in (0.5a , 0.7a)$).  Therefore, two possible solutions exist:
    \begin{enumerate}
    \item If neither bank liquidates any assets then $(q^{**},\bar q^{**}_1,\bar q^{**}_2,s^{**}_1,s^{**}_2) = (g(0),1,1,0,0)$ is an equilibrium solution. Indeed, set $s^{**}_1=s^{**}_2=0$ so that neither bank sells anything and determine the resulting prices $(q^{**},\bar q^{**}_1,\bar q^{**}_2)=(g(s^{**}_1+s^{**}_2),\hat f(s^{**}_1+s^{**}_2),\hat f(s^{**}_1+s^{**}_2))$; it is not difficult to verify that this is a clearing solution. 
    \item If both banks default and liquidate all their assets then an equilibrium solution is given by $(q^{**},\bar q^{**}_1,\bar q^{**}_2,s^{**}_1,s^{**}_2) = (g(2a),\hat f(2a),\hat f(2a), a,a)$. Again, set $s^{**}_1=s^{**}_2=a$ with resulting prices $(q^{**},\bar q^{**}_1,\bar q^{**}_2)=(g(s^{**}_1+s^{**}_2),\hat f(s^{**}_1+s^{**}_2),\hat f(s^{**}_1+s^{**}_2))$ and it is not difficult to verify that this is a clearing solution as well. 
    \end{enumerate}
\end{example}

%


\section{Main results}\label{sec:main}
We now concentrate our efforts into understanding when the above equilibrium is unique. In what follows we will investigate two specific sample functions. However, instead of specifying the inverse demand functions $\bar f_i$ directly, we derive them from a density function of limit order book together with some trading rules. Let this density be given by $f: \R_+ \to [0,1]$. Alternatively, this $f$ can be viewed as the price of the next infinitely small trade. We concentrate on two realistic examples of price constructions given the liquidations, i.e., market rules, to construct the price of the trade with functional forms $\bar f_i: \D \to (0,1]$ which provides the average price obtained by firm $i$ given the set of system liquidations.
\begin{enumerate}
\item {\bf Volume Weighted Average Price (VWAP):} For $i=1, ..., n, ~\s\in\D$ set $\bar f_i(\s) = 1$, if $\s=0$, otherwise let 
$\bar f_i(\s) \define \frac{ \int_0^{\sum_{j = 1}^n s_j} f(\sigma) d\sigma}{\sum_{j = 1}^n s_j}$. Note that $\bar f_i (\s) = \bar f_j(\s)$ for $i,j \in \{1,2,...,n\}$. 
\item {\bf Limit Order Book Based Price (LOB):} For $i=1, ..., n ~\s\in\D$ set $$\bar f_i(\s) \define \ind_{\{ s_i=0\}} + \ind_{\{s_i>0\}} \frac{1}{s_i} \sum_{j = 1}^k \frac{1}{n - (j-1)} \int_{\sum_{l = 1}^{j-1} (n - (l-1)) (s_{[l]} - s_{[l-1]})}^{\sum_{l = 1}^{j} (n - (l-1)) (s_{[l]} - s_{[l-1]})} f(\sigma) d\sigma,$$ where $0 =: s_{[0]} \leq s_{[1]} \leq s_{[2]} \leq ... \leq s_{[n]}$ are the order statistics and $s_i = s_{[k]}$.
\end{enumerate}

Note that the VWAP example corresponds to how some exchanges calculate the closing price (e.g., in Mexico, India and Saudi Arabia\footnote{\url{research.ftserussell.com/products/downloads/Closing_Prices_Used_For_Index_Calculation.pdf}}). Therefore, given our assumption that this is an illiquid asset, this is a good representation of price paid by banks given the amounts of trades they (collectively) want to make. 
Whereas the LOB example is an example of how to price market trades all coming at the same time using an existing limit order trades already in the book. This is a very interesting and novel example, as in this case, different banks pay different prices.
As far as the authors are aware, this LOB construction has never previously been formulated.

Alternatively, these specific pricing functionals can be viewed as a limit as order sizes decrease to zero at different rates. VWAP can be viewed as the limit when all banks submit their orders at a rate proportional to the total desired liquidation; as such, every bank finishes trading at the same ``time'' and thus all banks obtain the same average price. In contrast, the LOB is the limit when all banks submit their orders at the same rate; as such, banks finish their transactions at different ``times'' based on the desired quantity of assets to be liquidated which generates heterogeneous prices for different trading strategies.
Therefore, though this model is static, these constructions allow us to approximate simple time dynamics.

The following assumptions are placed on the order book density function $f$:
\begin{assumption}\label{ass:idf}
Let $M \geq \sum_{i = 1}^n a_i$ be the total initial market capitalization of the illiquid asset.  
The order book density function $f: [0,M] \to (0,1]$ is strictly decreasing and twice continuously differentiable, with $f(0) = 1$. 
Additionally it will be assumed that the first derivative $f': [0,M] \to -\R_+$ is nondecreasing.  
\end{assumption}

\begin{remark}
Under Assumption \ref{ass:idf} both the LOB and VWAP functions $\bar f_i,~i=1, ...,n$ satisfy Assumption \ref{ass:bar-f-g}(1)-(2).
\begin{itemize}
\item Let $\bar f_i$ be the VWAP inverse demand function and let $i = 1,...,n$.  Continuity on $\D\backslash\{\bf{0}\}$ of $\bar f_i$ follows directly from the construction of the VWAP and, by the fundamental theorem of calculus, $\lim\limits_{\s\to\bf{0}} \bar f_i(\s)=1$.  Additionally, we can caclulate $\frac{\d}{\d s_k} \bar f_i(\s) = -\frac{1}{\hat s^2} \int_0^{\hat s} f(\sigma)d\sigma + \frac{1}{\hat s} f\(\hat s\) < -\frac{\hat s f(\hat s)}{\hat s^2} + \frac{f(\hat s)}{\hat s} = 0$ for $\hat s = \sum_{j = 1}^n s_j$ and any bank $k$.  Moreover, $\frac{\d}{\d s_i} (s_i \bar f_i(\s)) = \frac{1}{\hat s}\int_0^{\hat s} f(\sigma)d\sigma - \frac{s_i}{\hat s^2}\int_0^{\hat s}f(\sigma)d\sigma + \frac{s_i f(\hat s)}{\hat s} > (\hat s - s_i) \frac{\hat s f(\hat s)}{\hat s^2} + s_i \frac{f(\hat s)}{\hat s} = f(\hat s) > 0$ for $\hat s = \sum_{j = 1}^n s_j \geq s_i$.  Finally, it is also easily seen that $\frac{\d}{\d s_i^2} (s_i \bar f_i(\s)) < \frac{s_i}{\sum_{j = 1}^n s_j} f'\(\sum_{j = 1}^n s_j\) \leq 0$.
\item Let $\bar f_i$ be the LOB inverse demand function and let $i = 1,...,n$.  Continuity on $(0,a_i] \times \prod_{j \neq i} [0,a_j]$ of $\bar f_i$ follows directly from the construction of the LOB.  When $\tilde \s$ is on the boundary of $\D$, assume by renaming that w.l.o.g.\ that $\tilde s_1 = ... =\tilde s_k=0,$ where $1\le k\le n$. It is then easily seen that for $1\le i\le k$, we have that $\lim\limits_{ \s\to \tilde\s} \bar f_i(\s) = 1 =\bar f_i(\tilde \s) $, and for $k<i\le n$, $\lim\limits_{ \s\to \tilde\s} \bar f_i(\s) =\bar f_i(\tilde \s)$, and we conclude the continuity on $\D$, as desired.  
Additionally, $\bar f_i$ is non-increasing in $\D$.  Fix $\s\in\D$ and let $\hat\s^j = (\hat s_j,\s_{-j})\in\D$ differ from $\s$ only in the $j$th component; without loss of generality let $\hat s_j > s_j$.  If $s_j \geq s_i$ with $j\neq i$ then $\bar f_i(\s) = \bar f_i(\hat \s^j)$ by construction of LOB.  If $s_j < \hat s_j\le s_i$ then $\bar f_i(\s) > \bar f_i(\hat \s^j)$ because $f$ is decreasing; similarly, if $j = i$ then $\bar f_i(\s) > \bar f_i(\hat s^i)$ since $f$ is decreasing.  If $s_j < s_i < \hat s_j$ then the result follows from a combination of the previous two cases.
Finally, $s_i \bar f_i(\s) = \sum_{j = 1}^k \frac{1}{n-(j-1)} \int_{\sum_{l = 1}^{j-1}(n-(l-1))(s_{[l]}-s_{[l-1]})}^{\sum_{l = 1}^j (n-(l-1))(s_{[l]}-s_{[l-1]})} f(\sigma)d\sigma$ is strictly increasing and concave because $f$ is strictly positive and $f' \leq 0$.
\end{itemize}
\end{remark}

Throughout the remainder of this work we often wish to consider a comparison of vectors of $(q,\bar \q)$; this is accomplished in the usual way, i.e., $(q^1,\bar \q^1) \geq (q^2,\bar \q^2)$ if and only if $q^1 \geq q^2$ and $\bar q^1_i \geq \bar q^2_i$ for every $i = 1,...,n$.

Our next goal is to ultimately establish uniqueness-type properties of the Nash equilibrium.  In order to do so, similarly to \cite{bichuch2019optimization}, we consider the problem with fixed liquidation price(s) and the haircut value as described in~\eqref{eq:s*q}.  As opposed to Theorem~\ref{thm:exist} above, we first show that there exist unique Nash equilibrium liquidations for these fixed prices as shown in Proposition~\ref{prop:dsc} below, the proof of which is delayed until Appendix \ref{sec:proof1}.

\begin{proposition}\label{prop:dsc}
Let $\Qh := \left\{(q,\bar \q) \in (0,1] \times (0,1]^n \; | \; q < \bar q_i \; \forall i = 1,2,...,n\right\}$. 
Under VWAP or LOB structure and Assumption~\ref{ass:idf}, given $(q,\bar q_1, ..., \bar q_n) \in \Qh$ 
there exists a unique set of equilibrium liquidations $\bar \s(q,\bar q_1, ..., \bar q_n) = \FIX_{\bar\s \in \D} \s^*(\bar \s,q,\bar q_1,...,\bar q_n)$ to \eqref{eq:s*q}, i.e., $\bar s_i(q,\bar \q) = s_i^*(\bar \s_{-i}(q,\bar \q) , q , \bar \q)$ for every bank $i$.
\end{proposition}

From Example \ref{ex:counterexample} it is clear that uniqueness of the equilibrium does not hold without further assumptions. 
However, we show in Theorem~\ref{thm:tarski} that the set of all fixed point prices $(q^*,\bar\q^*)$ in the Nash equilibrium of \eqref{eq:s*q} is a lattice under a VWAP pricing scheme.
In contrast, we show in Theorem~\ref{thm:maximal} that the set of all fixed point liquidations $\s^*$ in the Nash equilibrium of \eqref{eq:s*} is a lattice under a LOB pricing scheme.
We wish to stress that while the uniqueness of Nash equilibrium is a very important question, arguably the more desired property is to be the ``best'' Nash equilibrium, i.e.\ the largest clearing payment vector of \cite{EN01}. For this result to hold, it is imperative to know that they have a lattice structure, which is exactly what is shown in Theorems~\ref{thm:tarski} and \ref{thm:maximal}. This result becomes even more important since there are no additional assumptions required beyond those already imposed for existence in Theorem \ref{thm:exist}. 
The proof of these theorems are presented in Appendices \ref{sec:proof2} and \ref{sec:proof2lob}.
\begin{theorem}\label{thm:tarski}
Under the VWAP structure and Assumptions~\ref{ass:bar-f-g}\eqref{ass:bar-f-g-3}-\eqref{ass:bar-f-g-4} and~\ref{ass:idf}, the set of clearing haircuts and prices is a lattice; in particular, there exists a greatest and least clearing haircut and set of clearing prices: 
$(q^\uparrow,\bar q_1^\uparrow, ..., \bar q_n^\uparrow)\ge (q^\downarrow,\bar q_1^\downarrow, ..., \bar q_n^\downarrow)$.
\end{theorem}
\begin{proof}[Sketch of proof]
Taking advantage of Proposition~\ref{prop:dsc}, we find that the sum $\sum_{i = 1}^n \bar s_i$ is monotonic in $(q,\bar q_1,...,\bar q_n)$.  Therefore we apply Tarski's fixed point theorem.  The details are provided in Appendix \ref{sec:proof2}.
\end{proof}

\begin{theorem}\label{thm:maximal}
Under the LOB structure and Assumptions~\ref{ass:bar-f-g}\eqref{ass:bar-f-g-3}-\eqref{ass:bar-f-g-4} and~\ref{ass:idf}, there exists a greatest and least clearing haircut and set of clearing prices: 
$(q^\uparrow,\bar q_1^\uparrow, ..., \bar q_n^\uparrow)\ge (q^\downarrow,\bar q_1^\downarrow, ..., \bar q_n^\downarrow)$.
Furthermore, the set of clearing liquidations $\{\bar\s(q,\bar\q) \; | \; (q,\bar\q) \text{ is equilibrium}\}$ is a lattice.
\end{theorem}
\begin{proof}[Sketch of proof]
We first show that the relaxed problem, without the lower bound constraint so that there are no forced liquidations, admits a unique clearing solution.  This is used to demonstrate that the clearing liquidations for~\eqref{eq:s*} are monotonic  and we apply Tarski's fixed point theorem.  The details are provided in Appendix \ref{sec:proof2lob}.
\end{proof}
\begin{remark}
Intriguingly, the VWAP structure admits a lattice of clearing prices whereas the LOB structure admits a lattice of clearing liquidations.  Importantly, it does not necessarily follow that the lattice of prices implies a lattice of liquidations or vice versa.
\end{remark}

Finally, we introduce additional assumptions and establish uniqueness of the equilibrium in Theorem \ref{thm:unique} below, the proof of which is delayed until the Appendix \ref{sec:proof3}.
For such a result we introduce a simplified notation, let $\d_x := \frac{\d}{\d x}$ denote the partial derivative operator with respect to some variable $x$.

\begin{definition}
We will say that bank $i\in\{1, ...,n\}$ is fundamentally solvent if it is able to cover its shortfall in any case, that is if
$
h_i\le a_i\bar f_i(\a),
$ 
where $\a=(a_1, ..., a_n)^\T.$
\end{definition}

\begin{remark}\label{remark:solvency}
If bank $i$ is fundamentally solvent then there is a \emph{feasible} solution to the maximization problem \eqref{eq:s*q}, provided $(q,\bar \q) = (g(\sum_{i = 1}^n s_i) , \bar \f(\s))$ for some $\s \in \D$, since the feasible region is non-empty. Indeed,
\begin{enumerate}
\item $\frac{h_i}{\bar q_i} \le a_i$ if and only if $h_i \le a_i \bar q_i$.
\item $\frac{h_i - a_i q}{\bar q_i - q} \le a_i$ if and only if $h_i \leq a_i \bar q_i$. 
\item $\frac{h_i - a_i q}{\bar q_i - q} \le \frac{h_i}{\bar q_i}$ if and only if $h_i \leq a_i \bar q_i$. 
\end{enumerate}
\end{remark}

\begin{theorem}\label{thm:unique}
Assume all banks are fundamentally solvent. 
Under VWAP or LOB structure and Assumptions~\ref{ass:bar-f-g}\eqref{ass:bar-f-g-3}-\eqref{ass:bar-f-g-4} and~\ref{ass:idf},  if additionally,  $-c M \(  c_1 f'(0) 
 \wedge g'(0) \) < \min\limits_{j} \min\limits_{\s \in \D} \left(\bar f_j(\s)-g(\sum_{i = 1}^n s_i)\right)$ with $c = 3, c_1=\frac{1}{2}$ and $c = n, c_1 = \frac{n}{2}$ in case of VWAP and LOB, respectively, then there exists a unique clearing haircut and set of actualized prices $(q^*,\bar \q^*)$.  
\end{theorem}
\begin{proof}[Sketch of proof]
The proof follows from the Banach fixed point theorem and is presented in Appendix \ref{sec:proof3}.
\end{proof}


\begin{remark}
At this point we wish to recall Example \ref{ex:counterexample} which highlights a case of non-uniqueness of the clearing solution.  In that two banks setting, neither bank is fundamentally solvent since, by construction, $a \bar f(2a) < h$.  This highlights the importance of the assumption that all banks are fundamentally solvent in Theorem \ref{thm:unique} 
for the uniqueness of the clearing prices.  
Of course, even though we do not have uniqueness of the clearing price, we do have monotonicity between the two proposed equilibria as implied by Theorem~\ref{thm:tarski}. 
\end{remark}

\begin{remark}\label{rem:r}
With the consideration of existence and uniqueness of the clearing solution, the sensitivity of the equilibrium liquidations and prices to the repo rate $r$ is of great interest.  This is studied mathematically in Appendix~\ref{sec:r}.  This problem is intimately related to the problem of studying lenders in the model under construction -- the lenders interact with the borrowers through the interest rate only. (Such a model with clearing interest rates so that the cumulative lent amount is equal to the borrowed sum is beyond the scope of this work.) Intuitively, we would expect as interest rates rise, more cash is being lent and, thus, is available for borrowing.  Therefore, by studying the sensitivity of the clearing solutions to the repo rate $r$, we can quantify the (first-order) impacts of modeling lenders on the clearing prices. 

Intuitively, we expect that as the repo rate rises, and borrowing becomes more expensive the liquidation of the illiquid asset increase. It also follows from here that, the higher the interest rate, the lower the terminal asset price. Alternatively, from  a regulator's perspective, if the goal is to limit the extent of the fire sales, it can be achieved by controlling the interest rates, as was done recently in September 2019, and was also used extensively during the 2008 financial crisis (see \cite{quinnstanding} and \cite{cecchetti2009crisis} respectively). We refer to the case studies in Section~\ref{sec:cs} for visualizations of this notion.
\end{remark}

\section{Comparative statics}\label{sec:cs}
Before considering specific examples, we will first introduce a consideration for the computation of the clearing prices $(q,\bar\q) = (g(\sum_{i = 1}^n \bar s_i(q,\bar\q)),\bar \f(\bar\s(q,\bar\q))).$ This approach will always converge to the maximal price in both the VWAP and the LOB settings due to Theorems ~\ref{thm:tarski} and \ref{thm:maximal}. 
Specifically, these are computed via Picard iterations beginning from $(q^0,\bar \q^0) := {\bf{1}}_{n+1}$.  However, $\bar\s(q,\bar q)$ will require consideration for computation itself due to its game theoretic construction.  As provided in Proposition~\ref{prop:dsc} these liquidations exist and are unique.  In fact, due to the construction of the problem as discussed in the proof of that proposition, we are able to apply the algorithm provided in~\cite{rosen65}.  This is summarized in Algorithm~2 of~\cite{bichuch2019optimization} for the VWAP setting.  We wish to note that in the LOB setting, the computation can be improved significantly via an iterative approach of determining the banks liquidating the fewest number of assets.
%

In this section we will consider two primary case studies.  The first is a consideration of the VWAP and LOB structures to determine their relative ordering, i.e., is one better than the other.  This is important from a mechanism design perspective as different markets consider the closing price using different rule sets.  The second case study we will consider is an implementation of European banking data to determine the impacts of interest rates and haircut functions on the clearing prices.

\subsection{Mechanism design}\label{sec:mechanism}
In this first case study, we will investigate two financial settings in detail in order to show that some system constructions find that VWAP has more total liquidations with less system-wide use of the repo markets than LOB, while other constructions have the reverse ordering.  In particular, we will first consider a system of $n$ identical banks and second a specific system of $n = 2$ banks only.

\subsubsection{Symmetric case study}\label{sec:symmetric}
Consider a system of $n\ge2$ identical banks.  Each of these banks has shortfall $h > 0$ and assets $a > 0$.  The prevailing repo rate is provided by $r \in (0,\frac{1}{3})$.  For the purposes of this example, consider the order book density $f(s) = 1 - \alpha s$ and haircut function $g(s) = \frac{1}{2} - \alpha s$ for $\alpha \in \(\frac{2r}{(1+r)(n+1)a} , \frac{1}{2na}\)$; notably these constructions satisfy Assumptions~\ref{ass:bar-f-g}\eqref{ass:bar-f-g-3}-\eqref{ass:bar-f-g-4} and~\ref{ass:idf} and taken so as to construct an example in which firms have a choice of behavior.  Consider now our two market mechanisms: VWAP and LOB.
\begin{enumerate}
\item \textbf{VWAP:} By construction $\bar f_i(\s) := 1 - \frac{\alpha}{2}\sum_{j = 1}^n s_j$ for every bank $i$ in the VWAP construction.  Additionally, we take advantage of the symmetric setup to conclude that all banks should follow the same strategy, i.e., $\s^{VWAP} = s^{VWAP} {\bf{1}}_n$ for some singleton $s^{VWAP} \in [0,a]$ and $\bar \q^{VWAP} = \bar q^{VWAP} {\bf{1}}_n$ for some singleton $\bar q^{VWAP} \in [0,1]$.  Consider game~\eqref{eq:s*q} for fixed values $(q,\bar q)$ with $q < \bar q$:
\begin{align*}
&s_i^*(q,\bar q {\bf{1}}_n) = \argmin_{s_i \in [0,a]} \left\{\frac{\alpha}{2}(1+r)\left(\sum_{j \neq i} s_j^*(q,\bar q {\bf{1}}_n) + s_i\right)s_i + r(h - s_i) \; \left| \; s_i \in \left[\frac{h - a q}{\bar q - q} , \frac{h}{\bar q}\right]\right.\right\}\\
&\qquad= \argmin_{s_i \in [0,a]} \left\{\frac{\alpha}{2}(1+r) s_i^2 + \left[\frac{\alpha}{2}(1+r)(n-1)s^{VWAP}(q,\bar q) - r\right]s_i + rh \; \left| \; s_i \in \left[\frac{h - a q}{\bar q - q} , \frac{h}{\bar q}\right]\right.\right\}\\
&\qquad= \frac{h - aq}{\bar q - q} \vee \left[\frac{r}{(1+r)\alpha} - \frac{n-1}{2}s^{VWAP}(q,\bar q)\right] \wedge \frac{h}{\bar q},
\end{align*}
if $h < a \bar q$ (and $s^{VWAP}(q,\bar q) = a$ if $h \geq a \bar q$).  In particular, this provides a single fixed point problem to find $s^{VWAP}(q,\bar q)$, i.e.,
\begin{align*}
&s^{VWAP}(q,\bar q) = \frac{h - aq}{\bar q - q} \vee \left[\frac{r}{(1+r)\alpha} - \frac{n-1}{2}s^{VWAP}(q,\bar q)\right] \wedge \frac{h}{\bar q},\\
&\qquad \Rightarrow \; s^{VWAP}(q,\bar q) = \frac{h - aq}{\bar q - q} \vee \left[\frac{2 r}{\alpha (1+r)(n+1)}\right] \wedge \frac{h}{\bar q}
\end{align*}
if $h < a \bar q$.
We wish to note that the existence of $s^{VWAP}(q,\bar q)$ justifies our choice of $\s^{VWAP} = s^{VWAP}{\bf{1}}_n$ as, due to Proposition~\ref{prop:dsc}, $\s^{VWAP}$ is unique and thus must equal $s^{VWAP}{\bf{1}}_n$.
Finally, it remains to find the equilibrium prices $(q^{VWAP},\bar q^{VWAP})$:
\begin{align*}
q^{VWAP} &= \begin{cases} 
   -\frac{1}{2} + \sqrt{1 - 2\alpha n h} &\text{if } h \in \H_1^{VWAP},\\ 
   \frac{1}{2} - \frac{2rn}{(1+r)(n+1)} &\text{if } h \in \H_2^{VWAP},\\ 
   1 - \alpha n a - \frac{1}{2}\sqrt{1 + 8\alpha n (h - a) + 4 (\alpha n a)^2} &\text{if } h \in \H_3^{VWAP},\\ 
   \frac{1}{2} - \alpha n a &\text{if } h \in \H_4^{VWAP}, \end{cases}\\ 
\end{align*}
 \begin{align*}
 \bar q^{VWAP} &= \begin{cases}
   \frac{1 + \sqrt{1 - 2\alpha n h}}{2} &\text{if } h \in \H_1^{VWAP},\\ %
   1 - \frac{rn}{(1+r)(n+1)} &\text{if } h \in \H_2^{VWAP},\\ %
   \frac{5}{4} - \frac{\alpha n a}{2} - \frac{1}{4}\sqrt{1 + 8\alpha n (h - a) + 4 (\alpha n a)^2} &\text{if } h \in \H_3^{VWAP},\\ %
   1 - \frac{\alpha}{2} n a &\text{if } h \in \H_4^{VWAP}, \end{cases} 
\end{align*}
with borrowing/liquidation regions
\begin{align*}
\H_1^{VWAP} &= \left[0 \; , \; \frac{2r}{\alpha(1+r)(n+1)}\left(1 - \frac{rn}{(1+r)(1+n)}\right)\right) ,\\
\H_2^{VWAP} &= \left[\frac{2r}{\alpha(1+r)(n+1)}\left(1 - \frac{rn}{(1+r)(1+n)}\right) \; , \right. \\
    &\qquad\qquad \left. \frac{2r}{\alpha(1+r)(n+1)}\left(\frac{1}{2}+\frac{rn}{(1+r)(n+1)}\right) + a\left(\frac{1}{2} - \frac{2rn}{(1+r)(n+1)}\right)\right),\\
\H_3^{VWAP} &= \left[\frac{2r}{\alpha(1+r)(n+1)}\left(\frac{1}{2}+\frac{rn}{(1+r)(n+1)}\right) + a\left(\frac{1}{2} - \frac{2rn}{(1+r)(n+1)}\right) \; , \; a\left(1 - \frac{\alpha}{2}na\right)\right),\\ 
\H_4^{VWAP} &= \left[a \left(1 - \frac{\alpha}{2} n a\right) \; , \; \infty\right).
\end{align*}
We wish to note that all square roots are well defined on the intervals on which they are considered.  Additionally, $q^{VWAP}$ and $\bar q^{VWAP}$ are continuous in $h$; as such the closures of the bounding intervals can be chosen arbitrarily. 
Though this setting does not satisfy the uniqueness conditions of Theorem~\ref{thm:unique}, the simplicity of the symmetric system still admits a unique clearing solution. In this case, the uniqueness condition of Theorem \ref{thm:unique} assuming $M=na$ becomes $3na\alpha <\frac{1}{2}$ and may be violated. 

\item \textbf{LOB:} By construction $\bar f_{[i]}(\s) := 1 - \frac{\alpha}{2s_{[i]}}\[\sum_{k = 1}^{i-1} s_{[k]}(2s_{[i]} - s_{[k]}) + (n - (i-1))s_{[i]}^2\]$ for every bank $[i]$ (i.e., the bank liquidating the $i^{th}$ most assets) in the LOB construction.  Additionally, we take advantage of the symmetric setup to conclude that all banks should follow the same strategy, i.e., $\s^{LOB} = s^{LOB} {\bf{1}}_n$ for some singleton $s^{LOB} \in [0,a]$ and $\bar \q^{LOB} = \bar q^{LOB} {\bf{1}}_n$ for some singleton $\bar q^{LOB} \in [0,1]$.  Consider game~\eqref{eq:s*q} for fixed values $(q,\bar q)$ with $q < \bar q$:
\begin{equation*}
\resizebox{.95\textwidth}{!}{$
\begin{aligned}[t]
s_i^*(q,\bar q {\bf{1}}_n) &= \argmin_{s_i \in [0,a]} \left\{\left.\begin{array}{l} \frac{\alpha}{2}(1+r)\[ns_i^2\ind_{\{s_i \leq s^{LOB}(q,\bar q)\}}\right. \\ \left. + ((n-1)s^{LOB}(q,\bar q)^2 + 2(n-1)s^{LOB}(q,\bar q)s_i + s_i^2)\ind_{\{s_i > s^{LOB}(q,\bar q)\}}\] \\ \qquad + r(h - s_i) \end{array} \; \right| \; s_i \in \left[\frac{h - a q}{\bar q - q} , \frac{h}{\bar q}\right] \right\}\\
&= \begin{cases}
    \frac{h - aq}{\bar q - q} \vee \left[\frac{r}{\alpha(1+r)n}\right] \wedge \frac{h}{\bar q} &\text{if } \frac{r}{\alpha(1+r)n} \leq s^{LOB}(q,\bar q),\\
    \frac{h - aq}{\bar q - q} \vee \left[\frac{r}{\alpha(1+r)} - (n-1)s^{LOB}(q,\bar q)\right] \wedge \frac{h}{\bar q} &\text{if } \frac{r}{\alpha(1+r)n} > s^{LOB}(q,\bar q) ,\end{cases}
\end{aligned}
$}
\end{equation*}
if $h < a \bar q$ (and $s^{LOB}(q,\bar q) = a$ if $h \geq a \bar q$).  In particular, this provides a single fixed point problem to find $s^{LOB}(q,\bar q)$, i.e., if $h < a \bar q$
\begin{align*}
&s^{LOB}(q,\bar q) = \begin{cases}
    \frac{h - aq}{\bar q - q} \vee \left[\frac{r}{\alpha(1+r)n}\right] \wedge \frac{h}{\bar q} &\text{if } \frac{r}{\alpha(1+r)n} \leq s^{LOB}(q,\bar q),\\
    \frac{h - aq}{\bar q - q} \vee \left[\frac{r}{\alpha(1+r)} - (n-1)s^{LOB}(q,\bar q)\right] \wedge \frac{h}{\bar q} &\text{if } \frac{r}{\alpha(1+r)n} > s^{LOB}(q,\bar q), \end{cases}\\
&\qquad \Rightarrow \; s^{LOB}(q,\bar q) = \frac{h - aq}{\bar q - q} \vee \left[\frac{r}{\alpha(1+r)n}\right] \wedge \frac{h}{\bar q}
\end{align*}
as both provided cases result in the same fixed point.
We wish to note that the existence of $s^{LOB}(q,\bar q)$ justifies our choice of $\s^{LOB} = s^{LOB}{\bf{1}}_n$ as, due to Proposition~\ref{prop:dsc}, $\s^{LOB}$ is unique and thus must equal $s^{LOB}{\bf{1}}_n$.
Finally, it remains to find the equilibrium prices $(q^{LOB},\bar q^{LOB})$:
\begin{align*}
q^{LOB} &= \begin{cases} 
   -\frac{1}{2} + \sqrt{1 - 2\alpha n h} &\text{if } h \in \H_1^{LOB},\\ 
   \frac{1}{2} - \frac{r}{1+r} &\text{if } h \in \H_2^{LOB},\\ 
   1 - \alpha n a - \frac{1}{2}\sqrt{1 + 8\alpha n (h - a) + 4 (\alpha n a)^2} &\text{if } h \in \H_3^{LOB},\\ 
   \frac{1}{2} - \alpha n a &\text{if } h \in \H_4^{LOB},\end{cases}\\ 
\bar q^{LOB} &= \begin{cases}
   \frac{1 + \sqrt{1 - 2\alpha n h}}{2} &\text{if } h \in \H_1^{LOB},\\ 
   1 - \frac{r}{2(1+r)} &\text{if } h \in \H_2^{LOB},\\ 
   \frac{5}{4} - \frac{\alpha n a}{2} - \frac{1}{4}\sqrt{1 + 8\alpha n (h - a) + 4 (\alpha n a)^2} &\text{if } h \in \H_3^{LOB},\\ 
   1 - \frac{\alpha}{2} n a &\text{if } h \in \H_4^{LOB},\end{cases} 
\end{align*}
with borrowing/liquidation regions
\begin{align*}
\H_1^{LOB} &= \left[0 \; , \; \frac{r}{\alpha(1+r)n}\left(1 - \frac{r}{2(1+r)}\right)\right), \\
\H_2^{LOB} &= \left[\frac{r}{\alpha(1+r)n}\left(1 - \frac{r}{2(1+r)}\right) \; , \; \frac{r}{2\alpha(1+r)n}\left(1 + \frac{r}{1+r}\right) + a\left(\frac{1}{2} - \frac{r}{1+r}\right)\right) ,\\
\H_3^{LOB} &= \left[\frac{r}{2\alpha(1+r)n}\left(1 + \frac{r}{1+r}\right) + a\left(\frac{1}{2} - \frac{r}{1+r}\right) \; , \; a\left(1-\frac{\alpha}{2}na\right)\right) ,\\
\H_4^{LOB} &= \left[a \left(1 - \frac{\alpha}{2} n a\right) \; , \; \infty\right).
\end{align*}
We wish to note, as with the VWAP case above, that all square roots are well defined on the intervals on which they are considered.  Additionally, $q^{LOB}$ and $\bar q^{LOB}$ are continuous in $h$; as such the closures of the bounding intervals can be chosen arbitrarily. 
Though this setting does not satisfy the uniqueness conditions of Theorem~\ref{thm:unique}, the simplicity of the symmetric system still admits a unique clearing solution.  
In this case, the uniqueness condition of Theorem \ref{thm:unique} assuming $M=na$, and using the fact that $n\ge2$, becomes $n^3a\alpha <1$ and may be violated. 
\end{enumerate}

Notably, $s^{VWAP}(q,\bar q) \geq s^{LOB}(q,\bar q)$ for any choice of $(q,\bar q)$ by construction.  In fact, if there exist $n \geq 2$ banks, then this inequality is strict at equilibrium on $\H_2^{VWAP} \cap \H_2^{LOB}$, i.e.,
\begin{equation*}
\resizebox{1\textwidth}{!}{$
h \in \left(\frac{r}{\alpha(1+r)n}\left(1-\frac{r}{2(1+r)}\right) \; , \; \frac{2r}{\alpha(1+r)(n+1)}\left(\frac{1}{2}+\frac{rn}{(1+r)(n+1)}\right) + a\left(\frac{1}{2} - \frac{2rn}{(1+r)(n+1)}\right)\right).
$}
\end{equation*}

In contrast, by construction of the order book density $f$, the borrowing by each firm at equilibrium (and therefore total system wide borrowing) is smaller under the VWAP framework than the LOB framework, i.e., $h - s^{VWAP}(q^{VWAP},\bar q^{VWAP})\bar q^{VWAP} \leq h - s^{LOB}(q^{LOB},\bar q^{LOB}) \bar q^{LOB}$, with strict ordering on the same interval as given above.

\subsubsection{A counterexample to the symmetric ordering}\label{sec:counterexample}
In contrast to the symmetric system above, we now wish to consider a system in which the VWAP setting results in fewer liquidations and more borrowing than the LOB framework.
To do this, let's consider a simple heterogeneous $n = 2$ bank setting with $r = 0.01$, $\a = (1,2)$, and $\h = (0.3,1.2)$.  For this example consider the same order book density function $f(s) = 1 - \alpha s$ and haircut function $g(s) = \frac{1}{2} - \alpha s$, but with the specific price impact parameter $\alpha = 0.05$. 
With this construction, the clearing liquidations and prices can be determined numerically to be
\begin{itemize}
\item $\s^{VWAP} = (0 , 0.4853)$ with $q^{VWAP} = 0.4757$ and $\bar \q^{VWAP} = (0.9879 , 0.9879).$
\item $\s^{LOB} = (0.0990 , 0.5080)$ with $q^{LOB} = 0.4696$ and $\bar \q^{LOB} = (0.9950 , 0.9828).$
\end{itemize}
As desired at the beginning of this example, total liquidations are less for both banks (i.e., $\s^{VWAP} < \s^{LOB}$), but borrowing by both banks has the opposite order (i.e., $h_i - s_i^{VWAP} \bar q_i^{VWAP} > h_i - s_i^{LOB} \bar q_i^{LOB},~i=1,2.$).  This is the opposite order from the symmetric case study considered above; as such, there is no consistent order between the VWAP and LOB settings that can be determined.  

\subsubsection{Discussion}
As shown in the prior two examples, there is no consistent ordering between the VWAP and LOB settings.  For symmetric systems and, more generally, systems close to symmetric, if the market regulators wish to promote borrowing over liquidations, then the LOB framework is preferable; however, for certain heterogeneous systems, the VWAP framework may be preferable to that same regulator.  As such, the use of stress testing of different market mechanisms is of the paramount importance in order to determine the optimal market mechanism.

We wish to make one final consideration on the comparison of the VWAP and LOB frameworks.  We conjecture that the distinction between the two setting occurs only if some bank is both liquidating and borrowing.  Most prior works, e.g.,~\cite{AFM16}, consider only the situation in which firms can only liquidate in order to cover their liabilities.  Without borrowing allowed, the VWAP and LOB frameworks will always coincide at the aggregate level.  As such the mechanism choice of $\bar \f$ is irrelevant when considered in the standard literature (which is written in a VWAP style manner).

\subsection{EBA case study}
We conclude this work with a consideration of financial system calibrated to 2011 European banking data.  This stress test data has been utilized in numerous prior studies for studying interbank liability networks (e.g., \cite{GV16,CLY14,feinstein2017multilayer}).  We will calibrate and utilize this EBA dataset in much the same way as in \cite{bichuch2019optimization}, i.e., to have a more realistic system but one that still requires heuristics and, as such, is for demonstration purposes only.

As a stylized bank balance sheet, we will consider two categories of assets: \emph{cash assets} $c_i$ and \emph{illiquid assets} $a_i$. We will additionally consider two categories of liabilities: \emph{external liabilities} $\bar p_i$ and \emph{capital} $C_i$. 
In order to determine these values, we calibrate the system as in~\cite{bichuch2019optimization} but ignoring all interbank obligations considered as cash so as to discount default contagion and focus solely on price-mediate contagion as discussed in the remainder of this work.  The total assets $T_i$ and capital $C_i$ are provided by this dataset directly for each bank $i$.  The external liabilities $\bar p_i = T_i - C_i$ are computed by balance sheet construction.  It remains to split the total assets into cash and illiquid assets; we make this split according to the tier 1 capital ratio $R_i$, i.e., $c_i = R_i T_i$ and $a_i = (1-R_i)T_i$.

In order to complete our model, we need to consider the remaining parameters of the system.  We set the market capitalization $M = \sum_{i = 1}^n a_i$ to be the total number of shares of the illiquid assets held by the banks.  For this example we consider the linear order book density function $f(s) = 1 - \alpha s$ and haircut function $g(s) = \frac{7}{10} - \alpha s$ for $\alpha = \frac{1}{300M}$ (i.e., a $0.30$ euro haircut is charged on top of the ``market price'' $f(s)$).  
By construction, this setting satisfies all conditions of Theorem~\ref{thm:unique}.
We will focus on the impacts of altering the interest rate environment in order to compare the VWAP and LOB settings.  This is undertaken in the prevailing low interest rate environment during the period from which this data is collected.  For this study, no external shock is applied to the financial system.

\begin{figure}[t]
\centering
\begin{subfigure}[t]{0.45\textwidth}
\centering
\includegraphics[width=\textwidth]{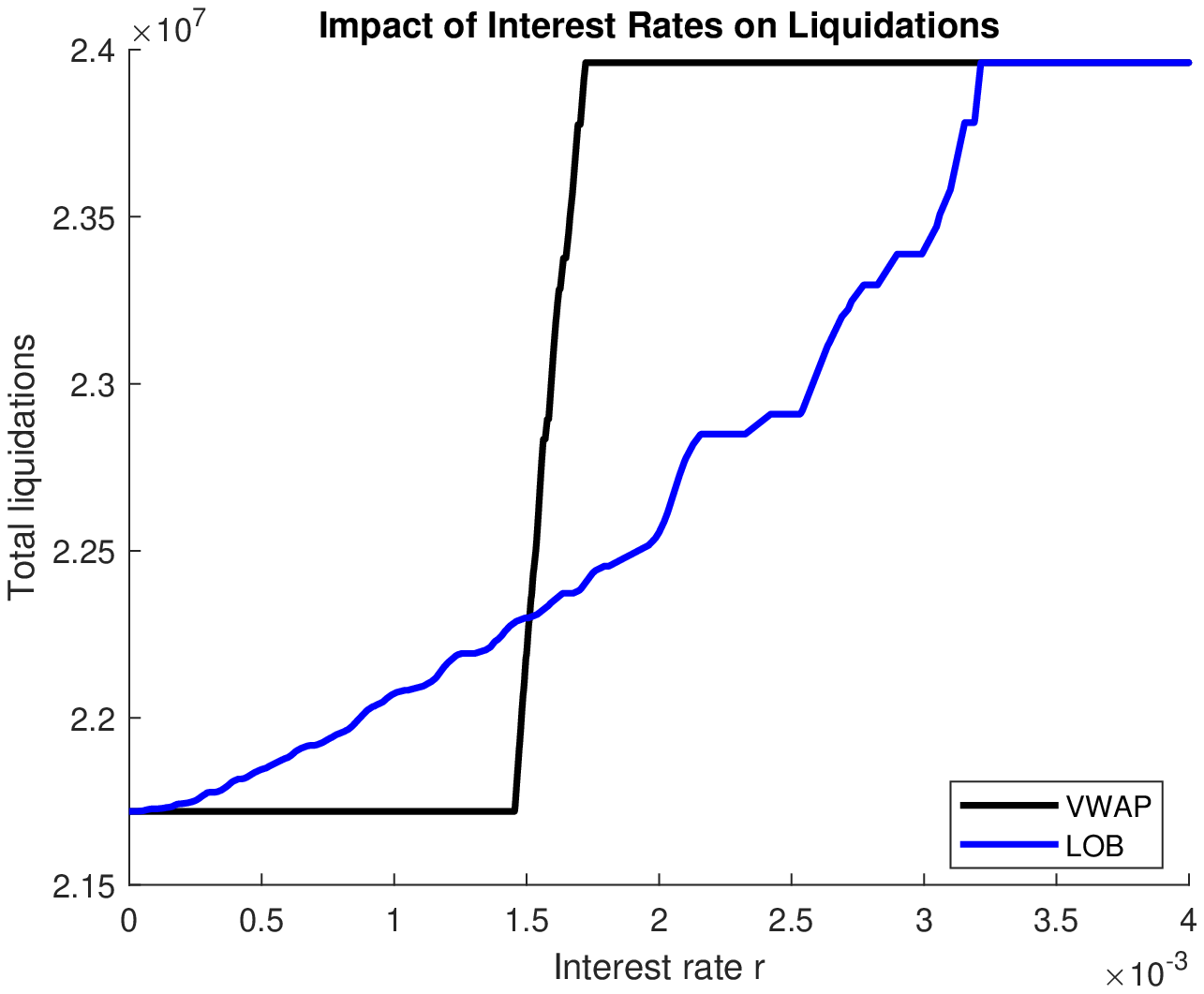}
\caption{Total liquidations}
\label{fig:EBA-liquidations}
\end{subfigure}
~
\begin{subfigure}[t]{0.45\textwidth}
\centering
\includegraphics[width=\textwidth]{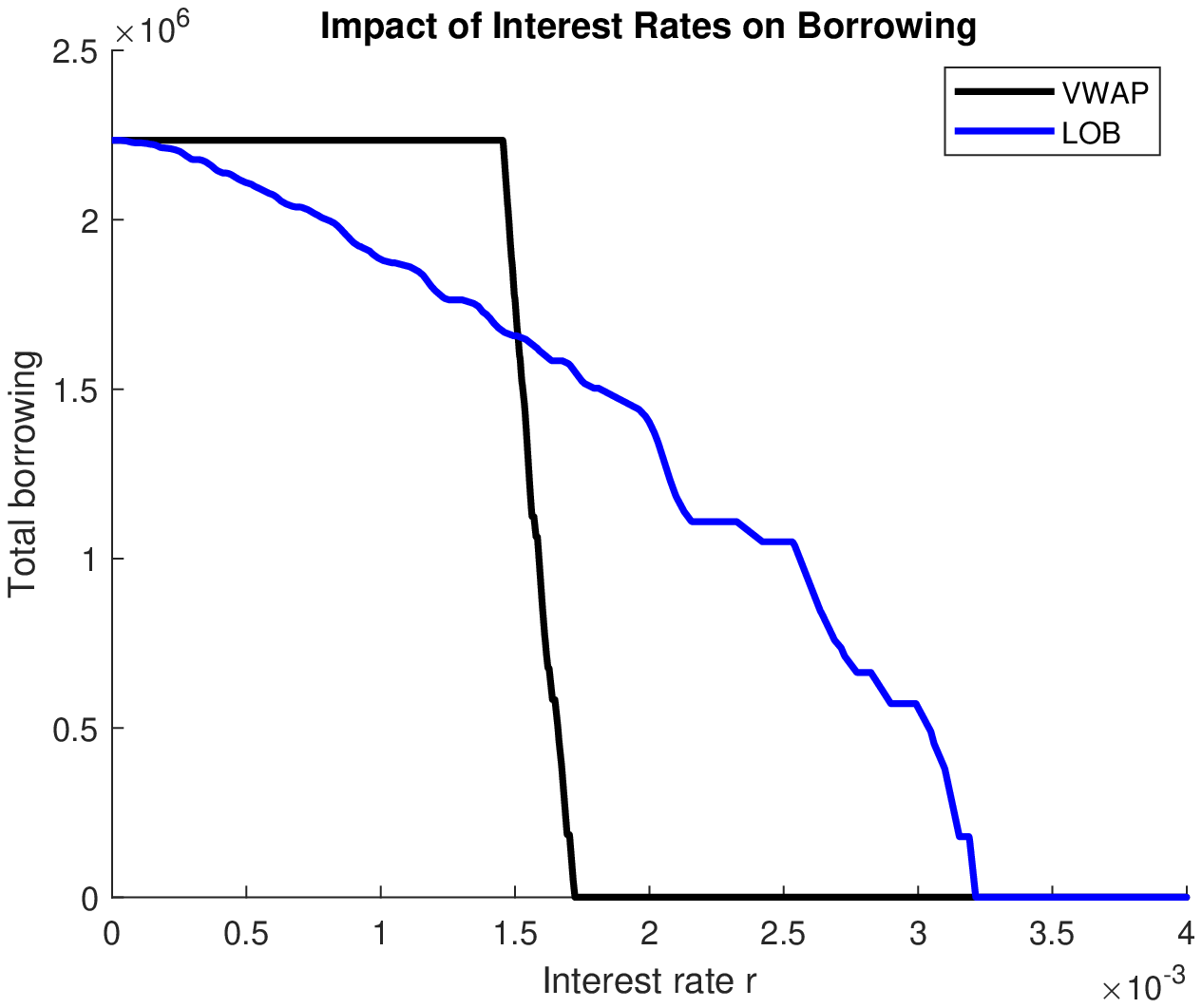}
\caption{Total borrowing}
\label{fig:EBA-borrowing}
\end{subfigure}
\caption{Summary statistics of the European banking sector's response to a changing interest rate environment.}
\label{fig:EBA}
\end{figure}
For our consideration, we compare the VWAP and LOB settings while varying the interest rate environment.  The results of varying the interest rate 
is displayed in Figure~\ref{fig:EBA}.  As expected, total liquidations (Figure~\ref{fig:EBA-liquidations}) increase as the interest rate increases, whereas the total borrowing (Figure~\ref{fig:EBA-borrowing}) is exactly the reverse of the total liquidations and, as such, is decreasing as the interest rate increases.  Notably, as discussed in the case studies of Section~\ref{sec:mechanism}, under some interest rate environments VWAP encourages more borrowing than LOB and vice versa under other interest rate environments.  
We find that, system-wide, there is less selling and more borrowing in the LOB setting for higher interest rates. 
Notably, the LOB setting results in a non-smooth response as a function of the interest rate $r$.  This results from the heterogeneous prices actualized by all banks; due to these varying prices, each bank switches strategies at varying interest rates.  This is in contrast to the VWAP setting in which, though the banks are heterogeneous, the strategies of the banks mostly overlap.  With this notion, it becomes clear that LOB provides greater flexibility for an intervention to control fire sales through the manipulation of interest rates.

\section{Conclusion}
In this work, we have considered a model of a system of banks that need to raise funds to cover their liquidity shortfalls. These firms decide on an optimal combination to raise the money through borrowing in a repo market and selling an illiquid asset in a fire-sale, with both the haircut and the fire-sale prices dependent on actions of other banks. We focused on two frameworks to determine the fire-sale prices: the volume weighted average price and a notion of the limit order book in order to capture notions of pricing dynamics in this, otherwise, static model.
We found sufficient conditions for existence maximal and minimal -- as well as uniqueness -- of the Nash equilibrium in this game.  Finally, we have compared the VWAP and the LOB settings analytically when the banks are identical and perform a numerical study using the 2011 EBA data.

\bibliographystyle{apalike}
\small{\bibliography{bibtex2}}

\newpage

\appendix
\section{Proofs of Section~\ref{sec:model}}
\subsection{Proof of Proposition~\ref{prop:nash-equiv}}\label{sec:proof-sec2}
\begin{proof}
First, assume that $\s^{**} \in\D$ is a Nash equilibrium of \eqref{eq:s*q}, and let $1\le i\le n$.  We will prove that $\s^{**}$ is a Nash equilibrium of \eqref{eq:s*} by considering all 4 cases that $s_i^{**}$ can be, i.e., at the unconstrained optimum, matching its shortfall through liquidations alone, borrowing as much as allowable, and being insolvent. 
\begin{itemize}
\item Let $s^{**}_i$ be equal to the \emph{unconstrained} optimal liquidations, i.e., 
$$\frac{\d}{\d s_i}\(  s_i \left(1 - \bar f_i(s_i,\s_{-i}^{**})\right) + r\left(h_i - s_i \bar f_i(s_i,\s_{-i}^{**})\right)\) \Big\vert_{s_i = s_i^{**}} =0.$$  
Given the behavior of all other banks $\s^{**}_{-i}$, from concavity of $s_i \mapsto s_i \bar f_i(s_i,\s^{**}_{-i})$, we have that $s^{**}_i$ is the optimal liquidation of bank $i$.  Therefore, $s^*_i = s^{**}_i$ since it is also feasible for \eqref{eq:s*} by construction of $q,\bar\q$.
\item Let $s^{**}_i = \frac{h_i}{\bar q_i}$ be such that bank $i$ matches its shortfall through liquidations alone.  Since this is optimal and the upper boundary for \eqref{eq:s*q}, it must follow that 
$$\frac{\d}{\d s_i}\(  s_i \left(1 - \bar f_i(s_i,\s_{-i}^{**})\right) + r\left(h_i - s_i \bar f_i(s_i,\s_{-i}^{**})\right)\) \Big\vert_{s_i = s_i^{**}} \le0.$$  
Therefore, given the behavior of all other banks $\s^{**}_{-i}$, from concavity of $s_i \mapsto s_i \bar f_i(s_i,\s^{**}_{-i})$ we have that $s^{**}_i$ is optimal for \eqref{eq:s*} if both $s^{**}_i$ is feasible, and $s^{**}_i + \epsilon$ is infeasible for any $\epsilon > 0$.  Feasibility of $s^{**}_i$ is trivial by construction of $q,\bar\q$. Invisibility of $s_i = s^{**}_i + \epsilon$, with $\epsilon > 0$, follows from that facts $s_i^{**}\bar q_i = h_i$ and that $s_i \mapsto s_i \bar f_i(s_i,\s^{**}_{-i})$ is strictly increasing by Assumption \ref{ass:bar-f-g}.
%
\item Let $s^{**}_i = \[\frac{h_i - a_i q}{\bar q_i - q}\]^+$ be such that bank $i$ is borrowing as much as allowable.  Since this is optimal and the lower boundary for \eqref{eq:s*q}, it must follow that 
$$\frac{\d}{\d s_i}\(  s_i \left(1 - \bar f_i(s_i,\s_{-i}^{**})\right) + r\left(h_i - s_i \bar f_i(s_i,\s_{-i}^{**})\right)\) \Big\vert_{s_i = s_i^{**}} \ge0.$$
Therefore, given the behavior of all other banks $\s^{**}_{-i}$, from concavity of $s_i \mapsto s_i \bar f_i(s_i,\s^{**}_{-i})$ we have that $s^{**}_i$ is optimal for \eqref{eq:s*} if both $s^{**}_i$ is feasible and $s^{**}_i - \epsilon$ is infeasible for any $\epsilon > 0$.  Feasibility of $s^{**}_i$ is trivial by construction of $q,\bar\q$. Additionally, for $\epsilon>0$, $s^{**}_i - \epsilon$ is infeasible because $s_i^{**} \bar f_i(\s^{**}) + (a_i - s_i^{**}) g\(\sum_{j = 1}^n s_j^{**}\) = s_i^{**} \bar q_i + (a_i - s_i^{**})q = h_i$ and $s_i\mapsto s_i \bar f_i(s_i,\s^{**}_{-i}) + (a_i - s_i)g\(s_i + \sum_{j \neq i} s_j^{**}\)$ is strictly increasing by Assumption \ref{ass:bar-f-g}. 
%
\item Let $s^{**}_i = a_i$ be such that bank $i$ is insolvent.  This situation occurs if, and only if, $a_i \bar q_i < h_i$.  It remains to show that every $s_i \in [0,a_i]$ is infeasible for \eqref{eq:s*} given the behavior $\s^{**}_{-i}$ of all other banks.  
Indeed, using the fact that $s_i\mapsto s_i \bar f_i(s_i,\s^{**}_{-i}) + (a_i - s_i)g\(s_i + \sum_{j \neq i} s_j^{**}\)$ is strictly increasing, we have that 
$s_i \bar f_i(s_i,\s^{**}_{-i}) + (a_i - s_i) g\(s_i + \sum_{j \neq i} s_j^{**}\) \leq a_i \bar f_i(a_i,\s^{**}_{-i}) = a_i \bar q_i < h_i$ and the lower boundary condition is violated making $s_i$ infeasible.
\end{itemize}
One degenerate case is if $s_i^{**} = a_i$ but the bank is not defaulting. In this case, we must have that $a_i \bar q_i \le h_i$ and  $a_i(\bar q_i - q)  \ge {h_i - a_i q},$ which together we have that $a_i \bar q_i =h_i$, and this is the only point in the constrained set. Clearly, this will also be a Nash equilibrium of \eqref{eq:s*}.

Vice versa, let $\s^* \in\D$ be a Nash equilibrium of \eqref{eq:s*}.  Since $s_i \mapsto s_i \bar q_i, s_i\mapsto s_i(\bar q_i -q)$ are both strictly increasing functions, we can prove that $\s^*$ is a Nash equilibrium of \eqref{eq:s*q} similarly. 
\end{proof}

\section{Proofs of Section~\ref{sec:main}}
\subsection{Proof of Proposition~\ref{prop:dsc}}\label{sec:proof1}
In both the VWAP and LOB settings, for a fixed $(q,\bar q_1, ..., \bar q_n) \in \Qh$ the existence of an equilibrium $\bar \s(q,\bar q_1, ..., \bar q_n)$ follows along the same steps as the proof of Theorem \ref{thm:exist}. We next show the uniqueness of $\bar \s(q,\bar q_1, ..., \bar q_n)$ by utilizing the results of~\cite{rosen65} on convex games.
\subsubsection{Volume weighted averaged price}
\begin{proof}
In this case, the uniqueness of $\bar \s(q,\bar q_1, ..., \bar q_n)$ follows from \cite{bichuch2019optimization}[Theorem 3.2], as soon as we verify that the assumptions of that theorem hold. 

Recall that $\bar f_i $ is independent of the index $i$.  Also, note that we can write $\bar f_i$ as a function of the total liquidation $s_{-i}$, where
\begin{align}
s_{-i} &= \sum\limits_{j\ne i,j=1}^ns_j,\label{eq:s-i}\\
\bar f_i (s_{i}, \s_{-i}) &=: \hat f(s_i + s_{-i}) =  \frac{1}{s_{i} + s_{-i}}  \int_0^{s_{i} + s_{-i} } f(s)ds.
\label{eq:VWAP-sum}
\end{align}
We assume that $f$ satisfies Assumption \ref{ass:idf}, and proceed to verify that \cite{bichuch2019optimization} [Assumption 2.1] is also satisfied by $\hat f$. Indeed, $\hat f'(s) =-\frac1{s^2}  \int_0^{s } f(u)du + \frac{f(s)}{s} < -\frac{sf(s)}{s^2} + \frac{f(s)}{s} =0,$ if $s>0$ and otherwise $\hat f'(0) = \frac12 f'(0)\le0$. 
It is also easily seen that $\frac{d^2}{ds^2}(s \hat f(s)) = f'(s) <0.$  

Lastly we need to show that $\hat f''\ge 0.$ First, $\hat f''(0)=\frac13f''(0)\ge0.$ We then calculate that 
$s^2\hat f''(s) = \frac2{s} \int_0^s f(u)du - 2 f(s) + s f'(s).$ Since $f$ is convex, we have that $f(s) - f(0) \le s f'(s),$ and thus it is sufficient to show that  $\frac2{s} \int_0^s f(u)du -f(0)-  f(s)\ge0.$
Using the fact that $f$ is convex, 
we have that $\lambda f(s_1) + (1-\lambda) f(s_2) \le f(\lambda s_1 + (1-\lambda)s_2), ~\lambda\in[0,1]$. Integrating over $\lambda\in[0,1]$ gives $\frac{f(s_1) + f(s_2)}{2} \le \frac1{s_2 -s_1}\int_{s_1}^{s_2} f(u)du,$ which gives the desired result. 
\end{proof}

\subsubsection{Limit order book}
\begin{proof}
Recall from \cite{rosen65} that for $\s\in\R^n$, the function $\s\mapsto H(\s;{\boldsymbol \rho})$ is diagonally strictly convex, if for some (fixed) ${\boldsymbol \rho}\in\R_{+}^n$ and for every $\s^0, \s^1\in \R^n,~\s^0 \neq \s^1$, we have $(\s^1 - \s^0)^\T \gamma(\s^0; {\boldsymbol \rho}) - (\s^1 - \s^0)^\T \gamma(\s^1; {\boldsymbol \rho}) <0$, where 
\begin{align}\gamma(\s;{\boldsymbol \rho})   =  \(\begin{array}{c}  \d_{s_1} H_1(\s;\rho_1)  \\  \vdots \\  \d_{s_n} H_n(\s;\rho_n)\end{array}\) = \(\begin{array}{c} \rho_1 \(1 - (1+r) f\(\sum_{\ell = 1}^{k_1} (n - (\ell-1))(s_{[\ell]} - s_{[\ell-1]})\)\) \\  \vdots \\ \rho_n \(1 - (1+r) f\(\sum_{\ell = 1}^{k_n} (n - (\ell-1))(s_{[\ell]} - s_{[\ell-1]})\)\) \end{array}\),
\end{align}
where $k_i$ is such that $s_{[k_i]} = s_i$.  
Additionally, \cite[Theorem 6]{rosen65} shows that a sufficient condition for $H$ to be diagonally strictly convex is if $\Gamma(\s;{\boldsymbol \rho}) + \Gamma(\s;{\boldsymbol \rho})^\T$ is a symmetric positive definite matrix for every $\s \in \R^n$ and some ${\boldsymbol \rho} \in \R_{+}^n$, where $\Gamma$ is the Jacobian matrix of $\gamma$ with respect to $\s$. 
Without loss of generality, for fixed value $\s$, assume $k_i = i$ for every bank.

Set $\rho_i = \frac1{1+r}$ then
\begin{align}
\[\Gamma(\s;{\boldsymbol \rho}) + (\Gamma(\s;{\boldsymbol \rho}))^\T\]_{ij} &= -(1 + \ind_{\{i=j\}} [2(n-i)+1]) f'\(\sum_{\ell = 1}^{i \vee j - 1} s_{\ell} + (n - (i \vee j - 1))s_{i \vee j}\).
\end{align}
Thus, in full matrix notation, we find
\begin{equation*}
\resizebox{1\textwidth}{!}{$
\begin{aligned}[t]
&\Gamma(\s;{\boldsymbol \rho}) + \Gamma(\s;{\boldsymbol \rho})^\T = A(\s) + \sum_{j = 1}^n B_j(\s),\\
&A(\s) = -\diag\([2(n-i)+1] f'\(\sum_{\ell = 1}^{i-1} s_\ell + (n - (i-1))s_i\)\),\\
&B_j(\s) = \begin{cases} \[f'\(\sum\limits_{\ell = 1}^{j-1} s_\ell + (n - (j-1))s_j\) - f'\(\sum\limits_{\ell = 1}^j s_\ell + (n - j) s_{j+1}\)\] \(\begin{array}{cc} {\bf{1}}_{j \times j} & {\bf{0}}_{j \times (n-j)} \\ {\bf{0}}_{(n-j) \times j} & {\bf{0}}_{(n-j) \times (n-j)} \end{array}\) &\text{if } j < n, \\ -f'\(\sum_{\ell = 1}^n s_\ell\) {\bf{1}}_{n \times n} &\text{if } j = n. \end{cases}
\end{aligned}
$}
\end{equation*}
For any liquidations $\s$, by construction, the matrix $A(\s)$ is positive definite and $B_j(\s)$ is positive semidefinite (by nondecreasing property of $f'$). The uniqueness of $\bar \s(q,\bar q_1, ..., \bar q_n)$ follows from \cite{rosen65}[Theorem 2].
\end{proof}

\subsection{Proof of Theorem~\ref{thm:tarski}}\label{sec:proof2}
Our goal is to apply Tarski's fixed point theorem, to do which, we need to prove that 
\begin{align}
\sum_{i=1}^n \bar s_i(q^\downarrow,\bar \q^\downarrow) \ge \sum_{i=1}^n \bar s_i(q^\uparrow,\bar \q^\uparrow), 
\label{eq:Tarski-ineq}
\end{align}
for $(q^\downarrow,\bar \q^\downarrow) \leq (q^\uparrow,\bar \q^\uparrow)$.  
\begin{proof}
Recall the definitions of $\hat f$ and $s_{-i}$ given in \eqref{eq:VWAP-sum} and \eqref{eq:s-i} respectively. Therefore, we can assume that $\bar q \define \bar q_1 =...=\bar q_n.$  First we note that for $i=1, ..., n$
\begin{equation}
\label{eq:s_i*-vwap}
\bar s_i(q,\bar q {\bf{1}}_n) = \begin{cases} a_i &\text{if } h_i \geq a_i \bar q, \\ \left(\frac{h_i - a_i q}{\bar q - q}\right)^+ \vee \left[s_i^0(\bar s_{-i}(q, \bar q {\bf{1}}_n)) \wedge \frac{h_i}{\bar q}\right] &\text{if } h_i < a_i \bar q, \end{cases}
\end{equation}
where $s_i^0$ is the solution to
\begin{align}
1 - (1+r)(\hat f(s_i^0 +\bar s_{-i}(q, \bar q {\bf{1}}_n) ) + s_i^0 \hat f'(s_i^0 + \bar s_{-i}(q, \bar q {\bf{1}}_n)) )= 0.
\label{eq:s-0}
\end{align}

With this construction we wish to note that bank $i$ is defaulting and has no other option but to liquidate all its assets if and only if $h_i \geq a_i \bar q$. Indeed, as noted previously in the body of this work, in the opposite case $h_i < a_i \bar q$, we have that:
\begin{enumerate}
\item $\frac{h_i}{\bar q} < a_i$ if and only if $h_i < a_i \bar q$.
\item $\frac{h_i - a_i q}{\bar q - q} < a_i$ if and only if $h_i < a_i \bar q$. 
\item $\frac{h_i - a_i q}{\bar q - q} < \frac{h_i}{\bar q}$ if and only if $h_i < a_i \bar q$. 
\end{enumerate}
Therefore, $\bar s_i(q,\bar q {\bf{1}}_n),~i=1, ...,n $ is well defined in \eqref{eq:s_i*-vwap}, and $\bar s_i(q,\bar q) <a_i$ in all those cases.

First consider the case when all banks keep at the same liquidation strategy, in other words the definition of $\bar s_i$ in \eqref{eq:s_i*-vwap} is equal to the same term (i.e., among $a_i$, $\frac{h_i}{\bar q}$, $\left(\frac{h_i - a_i q}{\bar q - q}\right)^+$, and $s_i^0$). Then for $i=1, ..., n$:
\begin{itemize}
\item If $\bar s_i = a_i$ then $\partial_q  \bar s_i = \partial_{\bar q}  \bar s_i= 0$.
\item If $\bar s_i = \frac{h_i}{\bar q}$, then $\partial_q  \bar s_i=0, \partial_{\bar q}\bar s_i <0$.  
\item If $\bar s_i = \left(\frac{h_i - a_i q}{\bar q - q}\right)^+$, first assume that $h_i\ge a_iq$, in addition to $h_i< a_i\bar q$. The former results in $\d_{\bar q} \bar s_i\le 0$, while it follows from the latter that $\d_q\bar s_i< 0$.  If, instead, $h_i < a_iq$ then $\d_{\bar q} \bar s_i = \d_q\bar s_i = 0$.
Note that we have also used our assumption that $\bar q > q$. 
\item The last case to consider is when $\bar s_i = s_i^0$. This is the most interesting case because 
$(s_i^0)' \in (-1,0]$ as shown in \cite{bichuch2019optimization}[Theorem 3.2], therefore the derivative has the opposite sign $\partial_{\bar q}s_i^0 = (s_i^0)'{\partial_{ \bar q}}s_{-i}\ge0$.
This case, requires a more careful analysis as follows.

Let $I_0$ be the set of banks $j=1, ...,n$, such that $\bar s_j = s^0_j$, then differentiating \eqref{eq:s-0} w.r.t. $\bar q$, and using the fact that $(s_i^0)' \in (-1,0]$, we see that
$$\d_{\bar q} {\vec{\bar s_{I_0}}} = -(\operatorname{diag}({\bf{1}}_{|I_0|} - \vec{c}) + \vec{c}{\bf{1}}_{|I_0|}^\T)^{-1} \vec{c} \sum_{j \not\in I_0} \d_{\bar q} \bar s_j,$$ for some $\vec{c} \in [0,1)^{|I_0|}$.  
First, we wish to show that $\operatorname{diag}({\bf{1}}_{|I_0|}-\vec{c}) + \vec{c}{\bf{1}}_{|I_0|}^\T$ is invertible:
\begin{align}
\operatorname{det}&\left(\operatorname{diag}({\bf{1}}_{|I_0|}-\vec{c}) + \vec{c}{\bf{1}}_{|I_0|}^\T\right) = \operatorname{det}\left(\begin{array}{cccc} 1 & c_{i_1} & \cdots & c_{i_1} \\ c_{i_2} & 1 & \cdots & c_{i_2} \\ \vdots & \vdots & \ddots & \vdots \\ c_{i_{|I_0|}} & c_{i_{|I_0|}} & \cdots & 1\end{array}\right)\\
&= \operatorname{det}\left(\begin{array}{cccc} 1 & -(1-c_{i_1}) & \cdots & -(1-c_{i_1}) \\ c_{i_2} & 1-c_{i_2} & \cdots & 0 \\ \vdots & \vdots & \ddots & \vdots \\ c_{i_{|I_0|}} & 0 & \cdots & 1-c_{i_{|I_0|}} \end{array}\right)\\
&= \left(1 + (1-c_{i_1}) \sum_{i \in I_0 \backslash \{i_1\}} \frac{c_i}{1-c_i}\right)\prod_{i \in I_0 \backslash \{i_1\}} (1-c_i)\\
&= \left(1 + (1-c_{i_1}) \sum_{i \in I_0} \frac{c_i}{1-c_i} - (1-c_{i_1}) \frac{c_{i_1}}{1-c_{i_1}}\right) \prod_{i \in I_0 \backslash \{i_1\}} (1-c_i)\\
&= \left(1 + \sum_{i \in I_0} \frac{c_i}{1-c_i}\right) \prod_{i \in I_0} (1-c_i),
\end{align}
where the 2nd line follows from subtracting the first column from every subsequent column and the 3rd line by using the Schur complement to determine the determinant.
Thus we find that $$\operatorname{det}\left(\operatorname{diag}({\bf{1}}_{|I_0|} - \vec{c}) + \vec{c}{\bf{1}}_{|I_0|}^\T\right) = \left(1 + \sum_{i \in I_0} \frac{c_i}{1 - c_i}\right) \prod_{i \in I_0} (1 - c_i) > 0.$$
\end{itemize}
Taking this all together:
\begin{align}
\sum_{i = 1}^n \d_{\bar q} \bar s_i &= \sum_{j \not\in I_0}\d_{\bar q} \partial \bar s_j - {\bf{1}}_{|I_0|}^\T (\operatorname{diag}({\bf{1}}_{|I_0|} - \vec{c}) + \vec{c}{\bf{1}}_{|I_0|}^\T)^{-1} \vec{c} \sum_{j \not\in I_0} \d_{\bar q}\bar s_j\\
&= \left(1 - {\bf{1}}_{|I_0|}^\T (\operatorname{diag}({\bf{1}}_{|I_0|} - \vec{c}) + \vec{c}{\bf{1}}_{|I_0|}^\T)^{-1} \vec{c}\right) \sum_{j \not\in I_0} \d_{\bar q}\bar s_j.
\label{eq:deriv-s*}
\end{align}
Moreover $1 - {\bf{1}}_{|I_0|}^\T (\operatorname{diag}({\bf{1}}_{|I_0|} - \vec{c}) + \vec{c}{\bf{1}}_{|I_0|}^\T)^{-1} \vec{c} \geq 0$, since:
\begin{align}
{\bf{1}}_{|I_0|}^\T (\operatorname{diag}({\bf{1}}_{|I_0|}-\vec{c}) + \vec{c}{\bf{1}}_{|I_0|}^\T)^{-1} \vec{c} &= {\bf{1}}_{|I_0|}^\T \left(\operatorname{diag}({\bf{1}}_{|I_0|}-\vec{c})^{-1} - \frac{\operatorname{diag}({\bf{1}}_{|I_0|}-\vec{c})^{-1} \vec{c} {\bf{1}}_{|I_0|}^\T \operatorname{diag}({\bf{1}}_{|I_0|}-\vec{c})^{-1}}{1 + {\bf{1}}_{|I_0|}^\T \operatorname{diag}({\bf{1}}_{|I_0|}-\vec{c})^{-1} \vec{c}}\right) \vec{c} \\
&= \sum_{i \in I_0} \frac{c_i}{1-c_i} - \frac{1}{1 + \sum_{i \in I_0} \frac{c_i}{1-c_i}}\sum_{i \in I_0}\sum_{j \in I_0} \frac{c_i c_j}{(1-c_i)(1-c_j)} \\
&= \left(1 - \frac{1}{1 + \sum_{i \in I_0} \frac{c_i}{1-c_i}}\sum_{j \in I_0} \frac{c_j}{1-c_j}\right)\sum_{i \in I_0} \frac{c_i}{1-c_i}\\
&= \left(1 + \sum_{i \in I_0} \frac{c_i}{1-c_i}\right)^{-1} \sum_{i \in I_0} \frac{c_i}{1-c_i} \leq 1,
\end{align}
where the first equality follows from the Sherman-Morrison matrix identity.

It now follows from \eqref{eq:deriv-s*} that $\sum_{i = 1}^n \d_{\bar q}\bar s_i \le 0$, as desired. The same calculation also shows that $\sum_{i = 1}^n \d_{ q}\bar s_i \le 0$, and therefore \eqref{eq:Tarski-ineq} holds.  

Finally, in the case, that some banks may switch liquidation strategies, we use the fact that the mappings $\bar s_i(\cdot, \cdot),s^0_i(\cdot)~i=1, ...,n$ are continuous.
If there is a switch in strategies for bank $i$ at some fixed point $q_0, \bar q_0$, then by continuity, of all the mappings in \eqref{eq:s_i*-vwap} it follows that both one sided derivatives
$\sum_{i = 1}^n \d^{+}_{\bar q_0}\bar s_i,\sum_{i = 1}^n \d^{-}_{\bar q_0}\bar s_i\le 0$. Therefore $\sum_{i=1}^n\bar s_i$ is decreasing in $\bar q$. Similar result also holds for $q$. We conclude that \eqref{eq:Tarski-ineq} holds.

\end{proof}

\subsection{Proof of Theorem~\ref{thm:maximal}}\label{sec:proof2lob}
Before providing the proof of this result, we wish to provide an auxiliary result.
Within this proof, we make extensive use of notation $s_i^0(\bar\s_{-i})$ to denote the (unique) solution to the first order condition
\begin{align}
1 - (1+r) f\(\sum_{j = 1}^n \ind_{\{\bar s_j < s_i\}}\bar s_j + (n - (k_i - 1))s_i\) = 0,
\label{eq:LOB-FOC}
\end{align}
where we recall that $k_i$ is (the minimal integer) such that $s_{[k_i]} = s_i$.
\begin{proposition}\label{prop:lob-order}
Consider the LOB structure and Assumption~\ref{ass:idf}.  Let $(q^*,\bar\q^*)$ be some set of equilibrium prices with associated liquidations $\s^* = \bar\s(q^*,\bar\q^*)$.
Define
\begin{align*}
I_0& \define \left\{ i\in\{1, ..., n\} \; | \; s_i^* = s_i^0(\s_{-i}^*)\right\},~I_U \define \left\{ i\in\{1, ..., n\} \; | \; s_i^* = \frac{h_i}{\bar q_i^*}  \right\},\\
I_a& \define\left \{ i\in\{1, ..., n\} \; | \; s_i^* = a_i\right\},~I_L \define \left\{ i\in\{1, ..., n\} \; | \; s_i^* =  \left(\frac{h_i - a_i q^*}{\bar q_i^* - q^*}\right)^+ \right\}.
\end{align*}
\begin{enumerate}
\item~$s_i^* \leq s_j^*$ if $h_i \leq h_j$ and $i,j \in I_0 \cup I_U$.
\item~$s_i^* \leq s_j^*$ if $i \in I_U$ and $j \in I_0$.
\item~$s_i^* \leq s_j^*$ for any $i \in I_0 \cup I_U$ and $j \in I_L$.
\end{enumerate}
\end{proposition}
\begin{proof}
\begin{enumerate}
\item Assume there exists some $i,j \in I_0 \cup I_U$ so that $s_i^* > s_j^*$ and $h_i \leq h_j$.  
First, assume $j \in I_U$, then $s_i^* \bar f_i(\s^*) > s_j^* \bar f_j(\s^*) = h_j \geq h_i$ which forms a contradiction.  The inequality holds by definitions of the LOB pricing $\bar f$, and the fact that $s_i^* > s_j^*$, and thus in order statistics the index of $ s_j^*$ is also smaller than that of $s_i^*$. 
Second, assume $j \in I_0$. Fundamentally, $s_i^0(\s_{-i}^*) = s_j^0(\s_{-j}^*)$ for any pair of banks $i,j$ by a simple contradiction argument of the uniqueness argument of Proposition~\ref{prop:dsc} as otherwise $\bar s_i(q^*,\bar q^*) = s_j^0(\s_{-j}^*)$ and $\bar s_j(q^*,\bar q^*) = s_i^0(\s_{-i}^*)$ is a distinct equilibrium liquidation solution. As $s_i^* = \min\{s_i^0(\s_{-i}^*) , \frac{h_i}{\bar q_i^*}\}$, we recover a contradiction.  
%
\item For any bank $i \in I_0 \cup I_U$ then $s_i^* = \min\{s_i^0(\s_{-i}^*),\frac{h_i}{\bar q_i^*}\}$.  As described above, $s_i^0(\s_{-i}^*) = s_j^0(\s_{-j}^*)$ for any pair of banks $i,j$.  Therefore the desired result immediately follows. 
\item By construction $s_j^* \geq s_j^0(\s_{-j}^*) = s_i^0(\s_{-i}^*) \geq s_i^*$ where the equality follows from the equivalence of $s^0$ across banks.
\end{enumerate}
\end{proof}

We will now consider the proof of Theorem~\ref{thm:maximal}.
\begin{proof}[Proof of Theorem~\ref{thm:maximal}]
Within this proof and without loss of generality, we will reorder the banks so that the shortfalls are in increasing order, i.e., $h_1 \leq h_2 \leq ... \leq h_n$.  
Where appropriate within this proof, consider the extension of the order book density function $f(s) := f(s \wedge M)$ to $\R_+$.

Consider, first, the relaxation of~\eqref{eq:s*} without the lower bound for every bank $i$:
\begin{equation}\label{eq:s*relax}
\begin{split}
&\argmin_{s_i \in [0,a_i]} \; s_i \left(1 - \bar f_i(s_i,\s_{-i})\right) + r\left(h_i - s_i \bar f_i(s_i,\s_{-i})\right)\\
&\qquad \text{s.t.} \quad s_i\bar f_i(s_i,\s_{-i}) \leq h_i.
\end{split}
\end{equation}
We will demonstrate that this relaxed problem has a unique equilibrium solution $\hat\s^*$.  First, there exists an equilibrium solution following the same arguments utilized in the proof of Theorem~\ref{thm:exist}.
Now, consider a fixed set $\hat I_a$ of banks liquidating their entire holdings in this relaxed problem with the modified relaxation:
\begin{equation*}
\hat s_i(\hat I_a) := \begin{cases} \argmin_{s_i \geq 0} \left\{s_i \left(1 - \bar f_i(s_i,\hat\s_{-i}(\hat I_a))\right) + r\left(h_i - s_i \bar f_i(s_i,\hat\s_{-i}(\hat I_a))\right) \; | \; s_i \bar f_i(s_i,\hat\s_{-i}(\hat I_a)) \leq h_i\right\} &\text{if } i \not\in \hat I_a \\ a_i &\text{if } i \in \hat I_a \end{cases}
\end{equation*}
Due to the logic of Proposition~\ref{prop:lob-order}, $\hat\s(\hat I_a)$ is unique for any $\hat I_a \subseteq \{1,...,n\}$ as it is constructed sequentially for all banks $i \not\in \hat I_a$.
By construction, $\hat\s^*$ is an equilibrium solution to~\eqref{eq:s*relax} if and only if $\hat\s^* = \hat\s(\hat I_a^*)$ for $I_a^* = \{i \; | \; \hat s_i^* = a_i\}$.
Assume now there exist two distinct equilibrium solutions $\hat\s^1 \neq \hat\s^2$ to the relaxed problem~\eqref{eq:s*relax}.  By the prior argument, this can occur if and only if these equilibria differ on the set of fully liquidating banks $\hat I_a^1 \neq \hat I_a^2$.  Take the minimal index $i = \min [(\hat I_a^1 \cup \hat I_a^2) \backslash (\hat I_a^1 \cap \hat I_a^2)]$; without loss of generality assume $i \in \hat I_a^1$.  
By the logic of Proposition~\ref{prop:lob-order}, and a sequential construction of $\hat\s(\cdot)$, it must follow that $\hat s_j^1 = \hat s_j^2$ for any $j < i$.  Letting $s_i^U(\s) := \inf\{s_i \in \R_+ \; | \; s_i \bar f_i(s_i,\s_{-i}) \geq h_i\}$.  We immediately conclude
\begin{align*}
\hat s_i^1 &= \min\{a_i \; , \; s_i^0(\hat s^1_1,..., \hat s^1_{i-1},a_{i+1},....,a_n) \; , \; s_i^U(\hat s^1_1,..., \hat s^1_{i-1},a_{i+1},....,a_n)\} \\
&= \min\{a_i \; , \; s_i^0(\hat s^2_1,..., \hat s^2_{i-1},a_{i+1},....,a_n) \; , \; s_i^U(\hat s^2_1,..., \hat s^2_{i-1},a_{i+1},....,a_n)\} = \hat s_i^2,
\end{align*}
which is a contradiction and thus it must follow that~\eqref{eq:s*relax} has a unique equilibrium solution.

Denote the unique equilibrium liquidations of the relaxed problem~\eqref{eq:s*relax} by $\hat\s^*$.  
We will now demonstrate that $\s^*$ is an equilibrium liquidation of~\eqref{eq:s*} if and only if $\s^* = L(\s^*)$ for $L: [\hat\s^*,\a] \to [\hat\s^*,\a]$ defined component-wise as 
$$L_i(\s^*) := \inf\{s_i \in [\hat s_i^*,a_i] \; | \; s_i \bar f_i(s_i,\s_{-i}^*) + (a_i - s_i)g(s_i + \sum_{j\neq i} s_j^*) \geq h_i \wedge (a_i \bar f_i(a_i,\s_{-i}^*))\}.$$
First, assume $\s^* = L(\s^*)$ and let $s_i^L(\s) := \inf\{s_i \in \R_+ \; | \; s_i \bar f_i(s_i,\s_{-i}) + (a_i - s_i) g(s_i + \sum_{j \neq i} s_j) \geq h_i \wedge (a_i \bar f_i(a_i,\s_{-i}))\}$.  Then $\s^* = L(\s^*) = \s^L(\s^*) \vee \hat\s^* = \s^L(\s^*) \vee \left[\s^0(\hat\s^*) \wedge \s^U(\hat\s^*) \wedge \a\right]$.  In fact, by Proposition~\ref{prop:lob-order} and construction of the LOB inverse demand function, $\s^0(\hat\s^*) = \s^0(\s^*)$ and $\s^U(\hat\s^*) = \s^U(\s^*)$.  Therefore, by construction of~\eqref{eq:s*}, $\s^*$ defines an equilibrium liquidation strategy.
Second, assume $\s^*$ is an equilibrium liquidation strategy of~\eqref{eq:s*}.  By the same argument as before, $\s^U(\hat\s^*) = \s^U(\s^*)$ and $\s^0(\hat\s^*) = \s^0(\s^*)$.  Therefore, $\s^* = \s^L(\s^*) \vee \left[\s^0(\s^*) \wedge \s^U(\s^*) \wedge \a\right] = \s^L(\s^*) \vee \hat\s^* = L(\s^*)$.
Finally, the desired result follows by Tarski's fixed point theory as $L: [\hat\s^*,\a] \to [\hat\s^*,\a]$ is nondecreasing by construction of the LOB and haircut functions.

\end{proof}

\subsection{Proof of Theorem~\ref{thm:unique}}\label{sec:proof3}
To show uniqueness, we consider the equilibrium prices, as a mapping of $(q^*,\bar \q^*)$ to liquidating positions of banks $\bar \s(q^*,\bar \q^*)$, and then to the resulting prices, and show the uniqueness of a fixed point to this mapping. 
To simplify notation throughout this proof, let $\Q_i,~i\in\{1,n\},$ denote the set of attainable prices. The case of VWAP corresponds to $i=1$ and $\Q_1:= \left\{(g(s),\hat f(s)) \; | \; s \in [0,M]\right\}$ and the case of LOB corresponds to $i=1$, and  $\Q_n:= \left\{(g(\sum_{i = 1}^n s_i),\bar f_1(\s), ..., f_n(\s)) \; | \; \s \in \D\right\}$.
Moreover, for convenience define
\begin{align}
I_0& \define \left\{ i\in\{1, ..., n\} \; | \; \bar s_i = s_i^0\right\},~I_U \define \left\{ i\in\{1, ..., n\} \; | \; \bar s_i = \frac{h_i}{\bar q_i}  \right\},\label{eq:def-I1}\\
I_a& \define\left \{ i\in\{1, ..., n\} \; | \; \bar s_i = a_i\right\},~I_L \define \left\{ i\in\{1, ..., n\} \; | \; \bar s_i =  \left(\frac{h_i - a_i q}{\bar q_i - q}\right)^+ \right\}.\label{eq:def-I2}\
\end{align}
As before, we divide the proof into the VWAP and LOB cases:
\subsubsection{Volume weighted average price}\label{sec:subsecproof3}
\begin{proof}
We first fix $\bar q = \hat f(s),~q = g(s)$ for some $s \in [0,M]$ (recall the definition of $\hat f$ from~\eqref{eq:VWAP-sum}) and look for an equilibrium $\bar s_i(q,\bar q {\bf{1}}_n) = s_i^{*}(\sum_{j\ne i}\bar s_j(q,\bar q {\bf{1}}_n),q,\bar q{\bf{1}}_n)$ for all $i=1, ..., n$. That is for the modified Nash equilibrium given by \eqref{eq:s*q} and formulated explicitly in \eqref{eq:s_i*-vwap}.

The next goal is to show $(q,\bar q) \mapsto (\Phi(q,\bar q),\bar \Phi(q,\bar q)) =(g(\sum_{j = 1}^n \bar s_j(q,\bar q {\bf{1}}_n)), \hat f(\sum_{j=1}^n \bar s_j(q,\bar q {\bf{1}}_n) )),  $ is a contraction mapping. That is, our goal is to show that $\abs{ \bar \Phi(q^1,\bar q^1) - \bar \Phi(q^2,\bar q^2)} \leq \bar L \norm{(q^1,\bar q^1) -(q^2,\bar q^2)}_{\infty}$, and $\abs{ \Phi(q^1,\bar q^1) -  \Phi(q^2,\bar q^2)} \leq L \norm{(q^1,\bar q^1) -(q^2,\bar q^2)}_{\infty}$ with $L,\bar L < 1$  for any attainable set of prices $(q^1,\bar q^1),$ $(q^2,\bar q^2) \in \Q_1$.  Without loss of generality, for this proof we will assume $q^1 \leq q^2$; therefore $(q^1,\bar q^2) \in \Qh$ as well.

Indeed, with the convention that $0/0 = 0$:
\begin{align}
\frac{\abs{\bar \Phi(q^1,\bar q^1) - \bar \Phi(q^2,\bar q^2)}}{\norm{(q^1,\bar q^1) - (q^2,\bar q^2)}_{\infty}} &\leq \frac{\abs{\bar \Phi(q^1,\bar q^1) - \bar \Phi(q^1,\bar q^2)}}{\norm{(q^1,\bar q^1) - (q^2,\bar q^2)}_{\infty}} + \frac{\abs{\bar \Phi(q^1,\bar q^2) - \bar \Phi(q^2,\bar q^2)}}{\norm{(q^1,\bar q^1) - (q^2,\bar q^2)}_{\infty}} \\
&\leq \frac{\abs{\bar \Phi(q^1,\bar q^1) - \bar \Phi(q^1,\bar q^2)}}{\abs{\bar q^1 - \bar q^2}} + \frac{\abs{\bar \Phi(q^1,\bar q^2) - \bar \Phi(q^2,\bar q^2)}}{\abs{q^1 - q^2}} \\
&= \frac{1}{\abs{\bar q^1 - \bar q^2}}\abs{\hat f\left(\sum_{j = 1}^n \bar s_j(q^1,\bar q^1 {\bf{1}}_n)\right) - \hat f\left(\sum_{j = 1}^n \bar s_j(q^1,\bar q^2 {\bf{1}}_n)\right)}\\
    &\qquad + \frac{1}{\abs{q^1 - q^2}}\abs{\hat f\left(\sum_{j = 1}^n \bar s_j(q^1,\bar q^2 {\bf{1}}_n)\right) - \hat f\left(\sum_{j = 1}^n \bar s_j(q^2,\bar q^2 {\bf{1}}_n)\right)}\\
&\leq -\hat f'(0)\left(\max_{(q,\bar q) \in \Q_1} \abs{\sum_{j = 1}^n \d_{\bar q} \bar s_j(q,\bar q {\bf{1}}_n)} + \max_{(q,\bar q) \in \Q_1} \abs{\sum_{j = 1}^n \d_q \bar s_j(q,\bar q {\bf{1}}_n)}\right). \label{eq:contract1}
\end{align}
Similarly for $\Phi(q,\bar q)$. 
Thus to be a contraction mapping, it is sufficient to show that 
\begin{align}
-\hat f'(0) \(\max_{(q,\bar q)\in \Q_1 } \abs{ \sum_{j=1}^n\d_{\bar q}\bar s_j(q,\bar q {\bf{1}}_n)} +\max_{(q,\bar q)\in\Q_1 } \abs{ \sum_{j=1}^n\d_{q}\bar s_j(q,\bar q {\bf{1}}_n)}\)  <1,\\
-g'(0) \(\max_{(q,\bar q)\in\Q_1 } \abs{ \sum_{j=1}^n\d_{\bar q}\bar s_j(q,\bar q {\bf{1}}_n)} +\max_{(q,\bar q)\in \Q_1} \abs{ \sum_{j=1}^n\d_{q}\bar s_j(q,\bar q {\bf{1}}_n)}\)  <1.
\end{align}

In order to show this, consider the sensitivity of $\bar s(q,\bar q {\bf{1}}_n)$ with respect to $q,\bar q$. Recall the construction of $\bar s$ given by \eqref{eq:s_i*-vwap}. Recall the definitions of $I_U, I_L, I_0$ from \eqref{eq:def-I1} and \eqref{eq:def-I2}.
Assume that $a_i$, $\frac{h_i}{\bar q}$, $s_i^0(\sum_{j \neq i} \bar s_j(q,\bar q {\bf{1}}_n) )$, $\frac{h_i - a_i q}{\bar q - q}$ are all different for all $i=1,..., n$, so that together with the continuity of $s^0$ it follows that $\bar s$ is differentiable with respect to $q, \bar q$ and its derivatives for a given bank $i$ are given by 
\begin{align}
\d_{\bar q}\bar s_i(q,\bar q {\bf{1}}_n) = &\Bigg(-\ind_{\{i\in I_U\}} \frac{h_i}{\bar q^2} -\ind_{\{i\in I_L\}} \frac{h_i - a_i q}{(\bar q-q)^2} + (s_i^0)'(\sum_{j\ne i} \bar s_j(q,\bar q {\bf{1}}_n)) (\sum_{j\ne i} \d_{\bar q}\bar s_j(q,\bar q {\bf{1}}_n))  \ind_{\{i\in I_0\}}    \Bigg),
\\
\d_{q}\bar s_i(q,\bar q {\bf{1}}_n) = &\Bigg(\ind_{\{i\in I_L\}} \frac{h_i - a_i \bar q}{(\bar q-q)^2} + (s_i^0)'(\sum_{j\ne i} \bar s_j(q,\bar q {\bf{1}}_n)) (\sum_{j\ne i} \d_{q}\bar s_j(q,\bar q {\bf{1}}_n))  \ind_{\{i\in I_0\}}    \Bigg).\label{eq:s-bar-vwap}
\end{align}
Here, 
the derivative of the optimal liquidations ($s_i^0(s_{-i})$) can be found via implicit differentiation:
$
(s_i^0)'(s_{-i}) = -\frac{\hat f'(s_{-i} + s_i^0(s_{-i})) + s_i^0(s_{-i}) \hat f''(s_{-i} + s_i^0(s_{-i}))}{2\hat f'(s_{-i} + s_i^0(s_{-i})) + s_i^0(s_{-i}) \hat f''(s_{-i} + s_i^0(s_{-i}))}.
$
Therefore  $(s_i^0)'(s_{-i})\in (-1,0]$ for all banks $i$ such that $\bar s_i = s_i^0$ if $\hat f'(s) + s \hat f''(s) \le0$ for every $s \in [0,M]$. 

Solving the system \eqref{eq:s-bar-vwap}, it follows that 
\begin{align}
&\d_{\bar q}\bar \s(q,\bar q) \\
&= -\( I - \diag\( \[(s_i^0)'(\sum_{j\ne i} \bar s_j(q,\bar q {\bf{1}}_n)) (\sum_{j\ne i} \d_{\bar q}\bar s_j(q,\bar q {\bf{1}}_n))  \ind_{\{i\in I_0\}} \]_{i = 1,...,n}\)\({\bf{1}}_{n\times n} - I\)\)^{-1} \\
&\quad\times\(\diag\(\[\ind_{\{i\in I_U\}} \]_{i=1,...,n} \)\frac{\h}{\bar q^2} +\diag\(\[\ind_{\{i\in I_L\}} \]_{i=1,...,n} \)\frac{\h - q\a}{(\bar q - q)^2}\),~~~~~~~
\label{eq:bar-s'}\\
&\d_{q}\bar \s(q,\bar q) \\
&= \( I - \diag\( \[(s_i^0)'(\sum_{j\ne i} \bar s_j(q,\bar q {\bf{1}}_n)) (\sum_{j\ne i} \d_{q}\bar s_j(q,\bar q {\bf{1}}_n))  \ind_{i\in I_0\}}\]_{i = 1,...,n}\)\({\bf{1}}_{n\times n} - I\)\)^{-1} \\
&\quad\times\diag\(\[\ind_{i\in I_L\}} \]_{i=1,...,n} \)\frac{\h -  \bar q\a}{(\bar q - q)^2}.
\end{align}
Using the fact that $(s_i^0)'(s_{-i})\in (-1,0]$ for $i=1,..., n$ as follows from the sufficient assumption of the theorem, it thus follows that 
\begin{equation*}
\resizebox{1\textwidth}{!}{$
\begin{aligned}[t]
&\left|{\bf{1}}_n^\T \d_{\bar q}\bar \s(q,\bar q)\right| \label{eq:bar-s'-bnd}\\
&\le \max_{\vec d\in [0,1)^n} \left|{\bf{1}}_n^\T (I + \diag(\vec d)({\bf{1}}_{n\times n} - I))^{-1}\( \diag\(\[\ind_{\{d_i=0,i\in I_U\}}\]_{i=1,...,n} \)\frac{\h}{\bar q^2}+\diag\(\[\ind_{\{d_i=0,i\in I_L\}}\]_{i=1,...,n} \)\frac{\h-q\a}{(\bar q-q)^2}\)\right|.
\end{aligned}
$}
\end{equation*}
To compute this maximum, let $B(\vec d) := I + \diag(\vec d)({\bf{1}}_{n\times n} - I) = \diag\({\bf{1}}_n - \vec d\) + \vec d {\bf{1}}_n^\T$. 
By the Sherman-Morrison formula 
$B(\vec d)^{-1} = \diag\({\bf{1}}_n-\vec d\)^{-1} -\frac1{1+{\bf{1}}_n^\T \diag\({\bf{1}}_n-d\)^{-1} \vec d} \diag\({\bf{1}}_n-\vec d\)^{-1}\vec d {\bf{1}}_n^\T \diag\({\bf{1}}_n-\vec d\)^{-1}$. 
It now follows that for any $j = 1,...,n$
\begin{align}
\sum_{i = 1}^n \(B(\vec d)^{-1}\)_{ij} \ind_{\{d_j = 0\}} &= \frac{1}{1 + \sum_{k = 1}^n \frac{d_k}{1-d_k}} \(1 + \sum_{k = 1}^n \frac{d_k}{1-d_k} - \sum_{k \neq j} \frac{d_k}{1-d_k}\) \ind_{\{d_j = 0\}} = \frac{\ind_{\{d_j = 0\}}}{1 + \sum_{k = 1}^n \frac{d_k}{1-d_k}}. ~~~\label{eq:B-bnd}
\end{align}
Together with Remark \ref{remark:solvency} we conclude that 
\begin{align}
&\max_{(q,\bar q)\in \Q_1}\abs{{\bf{1}}_n^\T \d_{\bar q}\bar \s(q,\bar q)}  \le\max_{(q,\bar q)\in \Q_1,\vec d\in [0,1)^n}\Bigg\vert {\bf{1}}_n^\T B(\vec d)^{-1}\label{eq:deriv-bnd}\\
&\qquad\qquad\qquad\qquad\times \( \diag\(\[\ind_{\{d_i=0,i\in I_U\}}\]_{i=1,...,n} \)\frac{\h}{\bar q^2}+\diag\(\[\ind_{\{d_i=0,i\in I_L\}}\]_{i=1,...,n} \)\frac{\h-q \a}{(\bar q-q)^2}\) \Bigg\vert \\
&\le \max_{\bar q\in [\hat f(M),1]}\abs{{\bf{1}}_n^\T\frac{\frac{\h}{\bar q}\wedge \a} {\bar q}} + \max_{(q,\bar q)\in\Q_1}\abs{{\bf{1}}_n^\T\frac{\frac{\h-q \a}{\bar q-q}\wedge \a} {\bar q-q}}\le \max_{\bar q\in [\hat f(M),1]} \frac{\sum_{i=1}^n a_i }{\bar q} +\max_{(q,\bar q)\in\Q_1} \frac{\sum_{i=1}^n a_i }{\bar q-q} \\
&\le \frac{M}{\hat f(M)} + \frac{M}{\min_{s \in [0,M]} \left(\hat f(s) - g(s)\right)}.
\end{align}
Similarly,
\begin{align}
&\max_{(q,\bar q)\in \Q_1}\abs{{\bf{1}}_n^\T \d_{q}\bar \s(q,\bar q)}  \le\max_{(q,\bar q)\in \Q_1,\vec d\in [0,1)^n}\Bigg\vert {\bf{1}}_n^\T B(\vec d)^{-1}
 \diag\(\[\ind_{\{d_i=0,i\in I_L\}}\]_{i=1,...,n} \)\frac{\h-\bar q\a}{(\bar q-q)^2} \Bigg\vert \\
&\le  \max_{(q,\bar q)\in\Q_1}\abs{{\bf{1}}_n^\T\frac{\frac{\h-\bar q\a}{\bar q-q}} {\bar q-q}}\le \max_{(q,\bar q)\in\Q_1} \frac{\sum_{i=1}^n a_i }{\bar q-q} \le  \frac{M}{\min_{s \in [0,M]} \left(\hat f(s) - g(s)\right)},
\end{align}
where in the last inequality we have used that fact that $a_i\ge \frac{h_i-a_i q}{\bar q-q}\ge \frac{h_i-a_i  \bar q}{\bar q-q} = -a_i + \frac{h_i-a_i q}{\bar q-q}\ge -a_i. $
Recalling \eqref{eq:contract1}, we conclude that $(\Phi,\bar \Phi)$ is a contraction mapping if $-3 M (\hat f'(0)\wedge \bar g'(0)) < \min_{s \in [0,M]} \left(\hat f(s)-g(s)\right)$. Finally, it can be seen that $\hat f'(s) =\frac{f(s) - \hat f(s) }{s}$. Therefore, $\hat f'(0) = \frac12 f'(0).$  

Recall that it was assumed that $a_i, \frac{h_i}{\bar q} , \frac{h_i-a_i q}{\bar q-q},s_i^0(\sum_{j \neq i} \bar s_j(q,\bar q {\bf{1}}_n) )$ are all different. If this assumption is violated, say $  s_i^0(\sum_{j \neq i} \bar s_j(q,\bar q {\bf{1}}_n) )<\frac{h_i-a_i q}{\bar q-q}= \frac{h_i}{\bar q}$, then we need to consider one-sided derivatives. In that case, the derivative from the right $\d_{\bar q+} \bar s_i(q,\bar q {\bf{1}}_n) =-\frac{h}{\bar q^2}$, while the derivative from the left $\d_{\bar q-}\bar s_i(q,\bar q {\bf{1}}_n) = -\frac{h_i-a_i q}{(\bar q-q)^2}.$ In this case, both one-sided derivatives would satisfy \eqref{eq:deriv-bnd}. The other cases, can be treated similarly. 

\end{proof}

\subsubsection{Limit order book}
\begin{proof}
We first fix $\bar \q = \bar \f(\s),~q = g(\sum_{i = 1}^n s_i)$ for some $\s \in \D$ and look for an equilibrium $\bar s_i(q,\bar \q) = s_i^{*}(\bar \s_{-i}(q,\bar \q),q,\bar \q)$ which is explicitly provided by 
\begin{equation}
\label{eq:s_i*-lob}
\bar s_i(q,\bar \q) = \begin{cases} a_i &\text{if } h_i \geq a_i \bar q_i, \\ \left(\frac{h_i - a_i q}{\bar q_i - q}\right)^+ \vee \left[s_i^0(\bar\s_{-i}(q,\bar\q)) \wedge \frac{h_i}{\bar q_i}\right] &\text{if } h_i < a_i \bar q_i \end{cases}
\end{equation}
where $s_i^0(\bar \s_{-i})$ solves the first order condition 
\begin{align}
1 - (1+r) f\(\sum_{j = 1}^n \ind_{\{\bar s_j < s_i\}}\bar s_j + (n - (k_i - 1))s_i\) = 0,
\nonumber 
\end{align}
where we recall that $k_i$ is such that $s_{[k_i]} = s_i$.
For simplicity, we will continue to assume that $\bar q_1 \ge \bar q_2 \ge ... \ge \bar q_n$. 

The next goal is to show $(q,\bar \q) \mapsto (\Phi(q,\bar \q),\bar \Phiv(q,\bar \q)) =(g(\sum_{j = 1}^n \bar s_j(q,\bar \q)), \bar \f_1(\bar \s(q,\bar \q)))$ is a contraction mapping, i.e., to show that $\norm{ \bar \Phiv(q^1,\bar \q^1) - \bar \Phiv(q^2,\bar \q^2)}_{\infty} \le \bar L \norm{(q^1,\bar \q^1) -(q^2,\bar \q^2)}_{\infty}$ and $\abs{ \Phi(q^1,\bar \q^1) -  \Phi(q^2,\bar \q^2)} \le L \norm{(q^1,\bar \q^1) -(q^2,\bar \q^2)}_{\infty}$ with $L,\bar L < 1$ for any $(q^1,\bar\q^1),(q^2,\bar\q^2) \in \Q_n$.  Without loss of generality, for this proof we will assume $q^1 \leq q^2$; therefore $(q^1,\bar\q^2) \in \Qh$.

Indeed, with the convention that $0/0 = 0$, for any $1 \le j \le n$: 
\begin{align}
&\frac{\abs{\bar \Phi_j(q^1,\bar \q^1) - \bar \Phi_j(q^2,\bar \q^2)}}{\norm{(q^1,\bar \q^1) - (q^2,\bar \q^2)}_{\infty}} \leq \frac{\abs{\bar \Phi_j(q^1,\bar \q^1) - \bar \Phi_j(q^1,\bar \q^2)}}{\norm{(q^1,\bar \q^1) - (q^2,\bar \q^2)}_{\infty}} + \frac{\abs{\bar \Phi_j(q^1,\bar \q^2) - \bar \Phi_j(q^2,\bar \q^2)}}{\norm{(q^1,\bar \q^1) - (q^2,\bar \q^2)}_{\infty}} \\
&\qquad \leq \sum_{k = 1}^n \frac{\abs{\bar \Phi_j(q^1,\bar \q^2_{\{1,...,k-1\}},\bar \q^1_{\{k,...,n\}}) - \bar \Phi_j(q^1,\bar \q^2_{\{1,...,k\}},\bar \q^1_{\{k+1,...,n\}})}}{\norm{\bar \q^1 - \bar \q^2}_{\infty}} + \frac{\abs{\bar \Phi_j(q^1,\bar \q^2) - \bar \Phi_j(q^2,\bar \q^2)}}{\abs{q^1 - q^2}} \\
&\qquad \leq \sum_{k = 1}^n \frac{\abs{\bar \Phi_j(q^1,\bar \q^2_{\{1,...,k-1\}},\bar \q^1_{\{k,...,n\}}) - \bar \Phi_j(q^1,\bar \q^2_{\{1,...,k\}},\bar \q^1_{\{k+1,...,n\}})}}{\abs{\bar q^1_{k} - \bar q^2_{k}}} + \frac{\abs{\bar \Phi_j(q^1,\bar \q^2) - \bar \Phi_j(q^2,\bar \q^2)}}{\abs{q^1 - q^2}} \\
&\qquad = \sum_{k = 1}^n \frac{1}{\abs{\bar q^1_{k} - \bar q^2_{k}}} \abs{\bar f_j\left(\bar\s(q^1,\bar \q^2_{\{1,...,k-1\}},\bar \q^1_{\{k,...,n\}})\right) - \bar f_j\left(\bar\s(q^1,\bar \q^2_{\{1,...,k\}},\bar \q^1_{\{k+1,...,n\}})\right)} \\
&\qquad\qquad + \frac{1}{\abs{q^1 - q^2}} \abs{\bar f_j\left(\bar\s(q^1,\bar \q^2)\right) - \bar f_j\left(\bar\s(q^2,\bar \q^2)\right)} \\
&\qquad \leq \sum_{k = 1}^n \max_{(q,\bar\q) \in \Q_n} \abs{\sum_{i = 1}^n \inf_{\s\in\D^o}\d_{s_i}\bar f_j( \s)\d_{\bar q_k} \bar s_i(q,\bar \q)} + \max_{(q,\bar\q)\in \Q_n} \abs{\sum_{i = 1}^n  \inf_{\s\in\D^o}\d_{s_i}\bar f_j( \s)}\d_q \bar s_i(q,\bar \q).
\end{align}
Similarly for $\Phi(q,\bar \q)$. 
Thus to be a contraction mapping, it is sufficient to show that for every $j=1, ..., n$
\begin{align}
 \sum_{k=1}^n\max_{(q,\bar q)\in \Q_n} \abs{ \sum_{i=1}^n\inf_{\s\in\D^o}\d_{s_i}\bar f_j( \s) \d_{\bar q_k}\bar s_i(q,\bar \q)}+ \max_{(q,\bar q)\in \Q_n} \abs{ \sum_{i=1}^n\inf_{\s\in\D^o}\d_{s_i}\bar f_j( \s)\d_q\bar s_i(q,\bar \q)}  <1,\label{eq:contract2.1}\\
-g'(0) \left(\sum_{k=1}^n\max_{(q,\bar q)\in \Q_n} \abs{ \sum_{i=1}^n\d_{\bar q_k}\bar s_i(q,\bar \q)}- \max_{(q,\bar q)\in\Q_n } \abs{ \sum_{i=1}^n\d_q\bar s_i(q,\bar \q)} \right)  <1.
\label{eq:contract2.2}
\end{align}

In order to show this, consider the sensitivity of $\bar \s(q,\bar \q)$ with respect to $q,\bar \q$. 
Recall again the definitions of $I_U, I_L, I_0$ from \eqref{eq:def-I1} and \eqref{eq:def-I2}.
Assume that $a_i, \frac{h_i}{\bar q_i} , s_i^0(\bar \s_{-i}(q,\bar \q) ),$  $\frac{h_i - a_i q}{\bar q_i - q}$ are all different for all $i=1,..., n$. Note that for different $i$, some of these quantities may be equal, namely, we must have $s_i^0 =s_j^0$, since if there is a solution $s^0$, it is unique. Similar to the proof in Section \ref{sec:subsecproof3}, otherwise, one sided derivatives can be considered.
Together with the continuity of $s^0$ it follows that $\bar s_i$ is differentiable with respect to $q, \bar \q$ and its derivatives for a given bank $i$ are given by 
\begin{equation*}
\resizebox{1\textwidth}{!}{$
\begin{aligned}[t]
\d_{\bar q_k}\bar s_i(q,\bar \q) = &\Bigg(-\ind_{\{i=k,i\in I_U\}} \frac{h_i}{\bar q_i^2} -\ind_{\{i=k,i\in I_L\}} \frac{h_i - a_i q}{(\bar q_i-q)^2} + \nabla s_i^0 (\bar \s_{-i}(q,\bar \q))  \cdot 
\[ \d_{\bar q_k} \bar s_{j} (q,\bar \q) \ind_{\{\bar s_j < s_i^0\}}   \]_{j=1, ..., n, j\ne i}
\ind_{\{i\in I_0\}}    \Bigg),
\label{eq:s-bar-lob}\\
\d_q\bar s_i(q,\bar \q) = &\Bigg(\ind_{\{i\in I_L\}} \frac{h_i - a_i \bar q_i}{(\bar q_i-q)^2} +  \nabla s_i^0 (\bar \s_{-i}(q,\bar \q))  \cdot \d_q \bar \s_{-i} (q,\bar \q)   
\ind_{\{i\in I_0\}}    \Bigg).
\end{aligned}
$}
\end{equation*}
Here, 
the derivative of the optimal liquidations ($s_k^0(\s_{-i})$) can be found via implicit differentiation of $1 - (1+r) f\(\sum_{j = 1}^n \ind_{\{\bar s_j < s_i^0\}} \bar s_j + (n - k_{i+1})s_i\) = 0$ to be
$\d_{s_j} s_i^0(\s_{-i}) =  -\frac{ \ind_{\{\bar s_j < s_i^0\}} }{n - (k_i-1)}$.
Set $\bar q' =\min_{j,k}\d_{s_k} \bar f_j({\bf{0}}_n) <0.$ 
Recall that  $\frac{h_i}{\bar q_i^2}, \frac{h_i - a_i q}{(\bar q_i-q)^2} $$ \le \frac{a_i}{\min_{\s \in \D} \left(\bar f_j(\s) - g(\sum_{i = 1}^n s_i)\right)}.$ 
Thus, for any $i_0\in I_0$, we have that $\sum_{i\in I_0} \d_{\bar q_{j}} \bar s_i(q,\bar \q)>-  \ind_{\{\bar s_j < s_{i_0}^0\}} \sum\limits_{ k\in\{k \colon \bar s_k<s_{i_0}^0\} }  \frac{a_k}{\min_{\s \in \D} \left(\bar f_j(\s) - g(\sum_{i = 1}^n s_i)\right)}$, for $j=1, ... , n.$ 
Therefore, we have that $ \abs{\sum_{k=l}^m\d_{\bar q_j}\bar s_k(q,\bar \q)} \le \sum\limits_{ k\in\{k \colon \bar s_k\ne s_{i_0}^0\} } \frac{a_k}{\min_{\s \in \D} \left(\bar f_j(\s) - g(\sum_{i = 1}^n s_i)\right)} \le \frac{M}{\min_{\s \in \D} \left(\bar f_j(\s) - g(\sum_{i = 1}^n s_i)\right)},$
for any $1\le l\le m \le n.$ 

We conclude that for any $j=1, ..., n$, we have that
$$\sum_{k=1}^n\max_{(q,\bar q)\in \Q_n} \abs{ \sum_{i=1}^n\inf_{\s\in\D^o}\d_{s_i}\bar f_j( \s)\d_{\bar q_k}\bar s_i(q,\bar \q)} \le  \abs{\bar q'} \frac{n M}{\min_{\s \in \D} \left(\bar f_j(\s) - g(\sum_{i = 1}^n s_i)\right)}.$$
Similarly, since
$ \abs{ \sum_{i=1}^n \d_q\bar s_i(q,\bar \q)} \le   \frac{M}{\min_{\s \in \D,m} \left(\bar f_m(\s) - g(\sum_{i = 1}^n s_i)\right)},$ we get that
$$\max_{(q,\bar q)\in \Q_n} \abs{ \sum_{i=1}^n\inf_{\s\in\D^o}\d_{s_i}\bar f_j( \s)\d_q\bar s_i(q,\bar \q)} \le \abs{\bar q'} \frac{M}{\min_{\s \in \D} \left(\bar f_j(\s) - g(\sum_{i = 1}^n s_i)\right)}.$$ 
Recalling \eqref{eq:contract2.1} we conclude that $\bar \Phiv$ is a contraction mapping if 
$$
- nM \min_{i,j}\inf_{\s\in\D^o}\d_{s_i}\bar f_j( \s)    <\min_j \min_{\s \in \D} \left(\bar f_j(\s)-g(\sum_{i = 1}^n s_i)\right).
$$ 
Finally, a technical by a straightforward calculation reveals that $\inf_{\s\in\D^o}\d_{s_i}\bar f_j( \s) =\frac{n f'(0)}{2},$ and we conclude that the condition becomes $-\frac{n^2 f'(0)}{2}   <\min_j \min_{\s \in \D} \left(\bar f_j(\s)-g(\sum_{i = 1}^n s_i)\right).
$

Similar 
$- nM g'(0)   < \min_j \min_{\s \in \D} \left(\bar f_j(\s)-g(\sum_{i = 1}^n s_i)\right)$ ensures that $\Phi$ is a contraction mapping.

\end{proof}

\section{Sensitivity of the clearing solutions to interest rates $r$}\label{sec:r}
In this section we consider the assumptions of Theorem \ref{thm:unique} with the goal of investigating the sensitivity of the (unique) equilibrium to the interest rate $r$. 
This provides the first-order impacts of lenders on the clearing solutions as the lenders interact with the borrowers in the model proposed herein through the repo rate $r$ only.  Therefore, though we do not model potential lenders (i.e., banks with negative shortfall), we can make some conclusions on how the system behavior will change based on the amount of cash available to the lenders -- the more cash available the lower the repo rate $r$.
To simplify notation, for this section we write $\bar s_i := \bar s_i(q,\bar q {\bf{1}}_n)$ or $\bar s_i := \bar s_i(q,\bar \q)$ where the values of $(q,\bar q)$ and $(q,\bar \q)$ is clear from context for the VWAP and LOB settings respectively.
In the following, we derive $\d_r \bar \s$, the derivatives of the equilibrium liquidations w.r.t.\ $r$. We then provide conditions under which the system-wide total liquidations increase with increase of $r$. 


\subsection{Volume weighted average price}
Initially, as in the prior proofs, assume that for each $i=1, ..., n$, the possible solutions to the optimization for $\bar s_i$ from \eqref{eq:s_i*-vwap}, namely $a_i, \frac{h_i}{\bar q} , \frac{h_i-a_i q}{\bar q-q},s_i^0(\sum_{j \neq i} \bar s_j )$, are all different. We want to study $\d_{r} \bar s_i$ for $i=1, ..., n$. 
From the previous assumption it follows that
\begin{align}
\d_{r} \bar s_i &= \begin{cases} 0 &\text{if } i \in I_a, \\ -\frac{h_i}{\bar q^2} \d_r \bar q &\text{if } i \in I_U, \\ \(\frac{h_i - a_i \bar q}{(\bar q - q)^2}\)\d_r q - \(\frac{h_i - a_i q}{(\bar q - q)^2}\)\d_r \bar q &\text{if } i \in I_L, \\ \d_r s_i^0 &\text{if } i \in I_0\end{cases}
\end{align}
where $I_a,I_U,I_L,I_0$ were defined in \eqref{eq:def-I1} and \eqref{eq:def-I2}. 

Before continuing, we will consider $\d_r \bar s_i$ for $i \in I_0$.  By construction, we have
\begin{align}
&-( \hat f(s_i^0 + \sum_{j \neq i} \bar s_j) + s_i^0  \hat f'(s_i^0 +\sum_{j \neq i} \bar s_j)) - (1+r) ( 2  \hat f'(s_i^0 +\sum_{j \neq i} \bar s_j) + s_i^0  \hat f''(s_i^0 + \sum_{j \neq i} \bar s_j)) \d_{r} \bar s_i\\
&\qquad  - (1+r) \sum_{j\ne i}( \hat f'(s_i^0 + \sum_{j \neq i} \bar s_j) + s_i^0  \hat f''(s_i^0 +\sum_{j \neq i} \bar s_j))  \d_{r}\bar s_j=0.
\end{align}
Recall that every bank $i \in I_0$ will satisfy the same condition, i.e., $\d_r \bar s_i = \d_r \bar s_j$ for every $i,j \in I_0$.  For notational simplicity let $s^0 = s_i^0,\d_r s^0 = \d_r s_i^0$ for arbitrary $i \in I_0$.  Let $c = \frac{ \hat f'(|I_0|s^0 + \sum_{j \not\in I_0} \bar s_j) + s^0  \hat f''(|I_0|s^0 +\sum_{j \not\in I_0} \bar s_j)}{ 2\hat f'(|I_0|s^0 + \sum_{j \not\in I_0} \bar s_j) + s^0  \hat f''(|I_0|s^0 +\sum_{j \not\in I_0} \bar s_j)}$ and $d = -\frac{\hat f(|I_0|s^0 + \sum_{j \not\in I_0} \bar s_j) + s^0 \hat f'(|I_0|s^0 + \sum_{j \not\in I_0} \bar s_j)}{(1+r)(2 \hat f'(|I_0|s^0 + \sum_{j \not\in I_0} \bar s_j) + s^0 \hat f''(|I_0|s^0 + \sum_{j \not\in I_0} \bar s_j))}$. Recall that by our Assumption \ref{ass:bar-f-g}, $0\le c<1$ and $d > 0$. Therefore, it can be shown that
\begin{equation}
\d_r s^0 = \frac{d}{1 + c(|I_0|-1)} - \frac{c}{1 + c(|I_0|-1)}\sum_{j \not\in I_0} \d_r \bar s_j.
\end{equation}

We can now consider the joint sensitivity of the haircut $q$ and price $\bar q$ to interest rates:
\begin{align}
\d_r q &= \left[\sum_{i \in I_0} \d_r s^0 + \sum_{i \not\in I_0} \d_r \bar s_i\right]g'(\sum_{i = 1}^n \bar s_i)\\
    &= \left[\frac{|I_0|d}{1 + c(|I_0|-1)} + \frac{1 - c}{1 + c(|I_0|-1)}\sum_{j \not\in I_0} \d_r \bar s_i\right] g'(\sum_{i = 1}^n \bar s_i)\\
\d_r \bar q &= \left[\sum_{i \in I_0} \d_r s^0 + \sum_{i \not\in I_0} \d_r \bar s_i\right] \hat f'(\sum_{i = 1}^n \bar s_i)\\
    &= \left[\frac{|I_0|d}{1 + c(|I_0|-1)} + \frac{1 - c}{1 + c(|I_0|-1)}\sum_{j \not\in I_0} \d_r \bar s_i\right] \hat f'(\sum_{i = 1}^n \bar s_i).
\end{align}
To simplify notation, let $\tilde c = \frac{1 - c}{1 + c(|I_0|-1)}$ and $\tilde d = \frac{|I_0|d}{1 + c(|I_0|-1)}$.  Therefore
\begin{align}
\d_r q &= \left[\tilde d + \tilde c \left(\frac{h - a \bar q}{(\bar q - q)^2} [\ind_{\{i \in I_L\}}]_i- \frac{h}{\bar q^2}\right) \d_r q- \tilde c\left(\frac{h - a q}{(\bar q - q)^2} [\ind_{\{i \in I_L\}}]_i  [\ind_{\{i \in I_U\}}]_i\right)\d_r \bar q \right] g'(\sum_{i = 1}^n \bar s_i), \\
\d_r \bar q &= \left[\tilde d + \tilde c \left(\frac{h - a \bar q}{(\bar q - q)^2} [\ind_{\{i \in I_L\}}]_i- \frac{h}{\bar q^2}\right) \d_r q- \tilde c\left(\frac{h - a q}{(\bar q - q)^2} [\ind_{\{i \in I_L\}}]_i  [\ind_{\{i \in I_U\}}]_i\right)\d_r \bar q \right]\hat f'(\sum_{i = 1}^n \bar s_i).
\end{align}
That is, the sensitivity of the haircut and prices $(q,\bar q)$ w.r.t.\ the interest rate $r$ is the solution of a linear system
\begin{align}
\left(\begin{array}{c} \d_r q \\ \d_r \bar q \end{array}\right) &= \left[I - W\right]^{-1} \left(\begin{array}{c} g'(\sum_{i = 1}^n \bar s_i) \tilde d \\ \hat f'(\sum_{i = 1}^n \bar s_i) \end{array}\right)\\
&= \left[I + \frac{\left(\begin{array}{c} g'(\sum_{i = 1}^n \bar s_i) \\ \hat f'(\sum_{i = 1}^n \bar s_i) \end{array}\right)\left(\begin{array}{cc} \frac{h - a \bar q}{(\bar q - q)^2} [\ind_{\{i \in I_L\}}]_i & -\left[\frac{h - aq}{(\bar q - q)^2} [\ind_{\{i \in I_L\}}]_i + \frac{h}{\bar q^2} [\ind_{\{i \in I_U\}}]_i \right] \end{array}\right) \tilde c}{1 - \tilde c\left[\left(\frac{h - a \bar q}{(\bar q - q)^2} [\ind_{\{i \in I_L\}}]_i\right)g'(\sum_{i = 1}^n \bar s_i) - \left(\frac{h - a q}{(\bar q - q)^2} [\ind_{\{i \in I_L\}}]_i + \frac{h}{\bar q^2} [\ind_{\{i \in I_U\}}]_i\right)\hat f'(\sum_{i = 1}^n \bar s_i)\right]}\right] \\
&\quad\times\left(\begin{array}{c} g'(\sum_{i = 1}^n \bar s_i) \tilde d \\ \hat f'(\sum_{i = 1}^n \bar s_i) \end{array}\right),\\
W &= \left(\begin{array}{c} g'(\sum_{i = 1}^n \bar s_i) \\ \hat f'(\sum_{i = 1}^n \bar s_i) \end{array}\right)\left(\begin{array}{cc} \frac{h - a \bar q}{(\bar q - q)^2} [\ind_{\{i \in I_L\}}]_i & -\left[\frac{h - aq}{(\bar q - q)^2} [\ind_{\{i \in I_L\}}]_i + \frac{h}{\bar q^2} [\ind_{\{i \in I_U\}}]_i \right] \end{array}\right) \tilde c.
\end{align}

Moreover, it also follows that
\begin{align}
\d_r \sum_{i = 1}^n \bar s_i &= \frac{\d_r q}{g'(\sum_{i = 1}^n \bar s_i)} \\
&= 1+ \frac{\tilde c\left(\left(\frac{h - a \bar q}{(\bar q - q)^2} [\ind_{\{i \in I_L\}}]_i\right)g'(\sum_{i = 1}^n \bar s_i) - \left(\frac{h - a q}{(\bar q - q)^2} [\ind_{\{i \in I_L\}}]_i + \frac{h}{\bar q^2} [\ind_{\{i \in I_U\}}]_i\right)\hat f'(\sum_{i = 1}^n \bar s_i)\right) }{1 - \tilde c\left[\left(\frac{h - a \bar q}{(\bar q - q)^2} [\ind_{\{i \in I_L\}}]_i\right)g'(\sum_{i = 1}^n \bar s_i) - \left(\frac{h - a q}{(\bar q - q)^2} [\ind_{\{i \in I_L\}}]_i + \frac{h}{\bar q^2} [\ind_{\{i \in I_U\}}]_i\right)\hat f'(\sum_{i = 1}^n \bar s_i)\right]}\\
&=\frac1{1 - \tilde c\left[\left(\frac{h - a \bar q}{(\bar q - q)^2} [\ind_{\{i \in I_L\}}]_i\right)g'(\sum_{i = 1}^n \bar s_i) - \left(\frac{h - a q}{(\bar q - q)^2} [\ind_{\{i \in I_L\}}]_i + \frac{h}{\bar q^2} [\ind_{\{i \in I_U\}}]_i\right)\hat f'(\sum_{i = 1}^n \bar s_i)\right]}.
\end{align}
It follows that $\d_r \sum_{i = 1}^n \bar s_i>0$ if $\left(\frac{h - a \bar q}{(\bar q - q)^2} [\ind_{\{i \in I_L\}}]_i\right)g'(\sum_{i = 1}^n \bar s_i) - \left(\frac{h - a q}{(\bar q - q)^2} [\ind_{\{i \in I_L\}}]_i + \frac{h}{\bar q^2} [\ind_{\{i \in I_U\}}]_i\right)\hat f'(\sum_{i = 1}^n \bar s_i)<\frac1{\tilde c},$ which happens if, for example, $\hat f'$ is small enough.


\subsection{Limit order book}
Initially, again assume that for each $i=1, ..., n$, the possible solutions ($a_i,\frac{h_i}{\bar q_i},\frac{h_i - a_i q}{\bar q_i - q},s_i^0(\sum_{j \neq i} \bar s_j)$) to the optimization \eqref{eq:s_i*-lob} are all different. As in the VWAP case, we want to study $\d_r \bar s_i$ for $i\in\{1, ..., n\}.$ From the previous assumption it follows that
\begin{equation}
\d_r \bar s_i = \begin{cases} 0 &\text{if } i \in I_a, \\ 
-\frac{h_i}{\bar q_i^2} \d_r \bar q_i &\text{if } i \in I_U, \\ 
\left(\frac{h_i - a_i \bar q_i}{(\bar q_i - q)^2}\right) \d_r q - \left(\frac{h_i - a_i q}{(\bar q_i - q)^2}\right) \d_r \bar q_i &\text{if } i \in I_L ,\\ 
\d_r s_i^0 &\text{if } i \in I_0 \end{cases}
\end{equation}
where $I_a,I_U,I_L,I_0$ were defined in \eqref{eq:def-I1} and \eqref{eq:def-I2}. 

Recall $s_i^0(\bar s_{-i})$ solves the first order condition 
\begin{align}
1 - (1+r) f\(\sum_{j = 1}^n \ind_{\{\bar s_j < s_i\}}\bar s_j + (n - (k_i - 1))s_i\) = 0,
\end{align}
where $k_i$ is such that $s_{[k_i]} = s_i$. As noted in the proof of Theorem~\ref{thm:unique} in the LOB case, we have that $\bar s_i=s_i^0$ is unique, and independent of $i$.   In fact, $s_i^0 = s_j^0$ for every $i,j \in I_0$.  We will denote this common value as $s^0$.
If $\bar s_i = s^0$ then from implicit differentiation of \eqref{eq:LOB-FOC}, we get that 
\begin{align}
\d_{r} s^0 = -\frac{ f\(\sum_{j = 1}^n \ind_{\{\bar s_j < s^0\}}\bar s_j + (n - (|I_0 \cup I_L| - 1))s^0\) }{(1+r)  f'\(\sum_{j = 1}^n \ind_{\{\bar s_j < s^0\}}\bar s_j + (n - (|I_0 \cup I_L| - 1))s^0\)  (n - (|I_0 \cup I_L| - 1))} - \frac{\sum_{i \in I_U} \d_r \bar s_i}{n - (|I_0 \cup I_L| - 1)}.
\end{align}

Now we want to consider the case of $\d_r \bar s_i$ for $i \in I_U$.  Notably, $\bar s_i < s^0$ for $i \in I_U$ by construction (see \eqref{eq:s_i*-lob}).  Therefore for such banks, there is no change to the attained prices $\bar q_i$ by a change in the interest rate, i.e., $\d_r \bar q_i = 0$ for $i \in I_U$.  This allows us to simplify $\d_r s^0$.

We can now consider the joint sensitivity of the haircut $q$ and prices $\bar \q$ to interest rates:
\begin{equation*}
\resizebox{1\textwidth}{!}{$
\begin{aligned}[t]
\d_r q &= \left[\sum_{i \in I_0 \cup I_L} \d_r \bar s_i\right] g'(\sum_{i = 1}^n \bar s_i) \\
    &= \left[-\frac{|I_0|f\(\sum_{j = 1}^n \ind_{\{\bar s_j < s^0\}}\bar s_j + (n - (|I_0 \cup I_L| - 1))s^0\) }{(1+r)  f'\(\sum_{j = 1}^n \ind_{\{\bar s_j < s^0\}}\bar s_j + (n - (|I_0 \cup I_L| - 1))s^0\)  (n - (|I_0 \cup I_L| - 1))} + \sum_{i \in I_L} \d_r \bar s_i\right] g'(\sum_{i = 1}^n \bar s_i),\label{eq:d_r_q}\\
\d_r \bar q_i &= \begin{cases} 0 &\text{if } i \in I_a \cup I_U, \\ 
-\frac{f\(\sum_{j = 1}^n \ind_{\{\bar s_j < s^0\}}\bar s_j + (n - (|I_0 \cup I_L| - 1))s^0\) }{(1+r)  f'\(\sum_{j = 1}^n \ind_{\{\bar s_j < s^0\}}\bar s_j + (n - (|I_0 \cup I_L| - 1))s^0\)  (n - (|I_0 \cup I_L| - 1))} \d_{s_i} \bar f_i(\bar \s) &\text{if } i \in I_0, \\
 -\frac{|I_0| f\(\sum_{j = 1}^n \ind_{\{\bar s_j < s^0\}}\bar s_j + (n - (|I_0 \cup I_L| - 1))s^0\) }{(1+r)  f'\(\sum_{j = 1}^n \ind_{\{\bar s_j < s^0\}}\bar s_j + (n - (|I_0 \cup I_L| - 1))s^0\)  (n - (|I_0 \cup I_L| - 1))} \d_{s_{i_0}} \bar f_i(\bar \s) + \sum_{\substack{j \in I_L \\ \bar s_j \leq \bar s_i}} \left(\d_r \bar s_j\right) \d_{s_j} \bar f_i(\bar s) &\text{if } i \in I_L, \end{cases}
\end{aligned}
$}
\end{equation*}
for arbitrary $i_0 \in I_0$.

To simplify notation, let $\tilde c =-\frac{f\(\sum_{j = 1}^n \ind_{\{\bar s_j < s^0\}}\bar s_j + (n - (|I_0 \cup I_L| - 1))s^0\) }{(1+r)  f'\(\sum_{j = 1}^n \ind_{\{\bar s_j < s^0\}}\bar s_j + (n - (|I_0 \cup I_L| - 1))s^0\)  (n - (|I_0 \cup I_L| - 1))} $. 
It then follows that
\begin{align}
\left(\begin{array}{c} \d_r \bar{\vec  q} \\ \d_r  q \end{array}\right) = \vec W^{-1} \vec b,
\end{align}
where
\begin{align}
\vec W&= (w_{i,j})_{1\le i,j\le n+1},\\
\vec b &= (b_i)_{i=1, ...,n+1},\\
w_{i,j} &= \ind_{\{i=j<n+1\}} + \ind_{\{i,j\in I_L, \bar s_i\ge \bar s_j\}}   \frac{h_j-a_{j} q_j}{\(\bar q_j-q\)^2}   \d_{s_j}\bar f_i(\bar \s) -\ind_{\{i\in I_L, j=n+1\}}\sum_{k\in I_L, \bar s_k\le \bar s_i}  \frac{h_k-a_{k}\bar q_k}{\(\bar q_k-q\)^2}\d_{s_k}\bar f_i(\bar \s) , \\
w_{n+1, j} &= \ind_{\{j\in I_L\}} \frac{h_j-a_{j} q}{\(\bar q_j-q\)^2} g'+ \ind_{\{j =n+1\}} \(1-g'\sum_{i\in\I_L}\frac{h_i-a_{i}q}{\(\bar q_i-q\)^2} \),\\
b_{i} &= \tilde c\d_{s_i} \bar f_i(\bar \s) \ind_{\{i\in I_0\}} + \abs{I_0}\tilde c\d_{s_{i_0}} \bar f_i(\bar \s) \ind_{\{i\in I_L\}} + |I_0| \tilde c g'   \ind_{\{i=n+1\}}. 
\end{align}
We note, without loss of generality, that if we assume that for any $i\in I_a, j\in I_U, k\in I_0, l\in I_L,$ we have that $i<j<k<l$, and that for any $i,j\in I_L$, such that $i<j$ then $\bar s_i\le \bar s_j$, we then have that $\vec W $ is lower triangular, but has an addition of one full $n+1$ column. $\vec W$ is invertible, and we can find its inverse as follows: 
Note that $\vec W$ can be written as 
$$\vec W = \(\vec W_0+  (0,0, ..., 0,1)^\T \[ \ind_{\{j\in I_L\}} \frac{h_j-a_{j} q}{\(\bar q_j-q\)^2} g'- \ind_{\{j =n+1\}} g'\sum_{i\in\I_L}\frac{h_i-a_{i}q}{\(\bar q_i-q\)^2} \]_{j=1, ..., n+1}\)^\T,$$ 
where $\vec W_0 = \vec D\( I + \vec N\)$ with $$\vec D =  \diag\(\[ 1+ \ind_{\{j\in I_L\}}   \frac{h_j-a_{j} q_j}{\(\bar q_j-q\)^2}   \d_{s_j}\bar f_j(\bar \s)\]_{j=1, ...,n}  \)$$ is a diagonal matrix and a nilpotent matrix $$\vec N = \[ \(\ind_{\{i,j\in I_L, \bar s_i> \bar s_j\}}   \frac{h_j-a_{j} q_j}{\(\bar q_j-q\)^2}   \d_{s_j}\bar f_i(\bar \s)\)_{j,i}\]_{1\le i,j\le n+1}.$$
Note that $\vec N $ is such that $\vec N^{n+1}=0.$ Therefore, we have that $$\vec W_0^{-1} = \( I + \vec N\)^{-1} \vec D^{-1} = \(I - \vec N + \vec N^{2} + ... + (-1)^n\vec N^n\) \vec D^{-1}.$$ Finally,
\begin{align}
\vec W^{-1} = \(\vec W_0^{-1} - \frac{\vec W_0^{-1}  (0,0, ..., 0,1)^\T \[ \ind_{\{j\in I_L\}} \frac{h_j-a_{j} q}{\(\bar q_j-q\)^2} g'- \ind_{\{j =n+1\}} g'\sum_{i\in\I_L}\frac{h_i-a_{i}q}{\(\bar q_i-q\)^2} \]_{j=1, ..., n+1}  \vec W_0^{-1}}{1 +  \[ \ind_{\{j\in I_L\}} \frac{h_j-a_{j} q}{\(\bar q_j-q\)^2} g'- \ind_{\{j =n+1\}} g'\sum_{i\in\I_L}\frac{h_i-a_{i}q}{\(\bar q_i-q\)^2} \]_{j=1, ..., n+1} \vec W_0^{-1} (0,0, ..., 0,1)^\T}\)^\T.
\end{align}

To calculate $\d_r \sum_{i=1}^n \bar s_i$, recall that $\d_r \bar q_i = 0$ for $i \in I_U$, and from \eqref{eq:d_r_q} it follows that
\begin{align}
\d_r \sum_{i=1}^n \bar s_i = \sum_{i \in I_0 \cup I_L} \d_r \bar s_i = \frac{\d_r q}{g'(\sum_{i = 1}^n \bar s_i)} = \frac{(0,0, ..., 0,1)^\T \vec W^{-1} \vec b}{g'(\sum_{i = 1}^n \bar s_i)} . 
\end{align}
It follows that $\d_r \sum_{i = 1}^n \bar s_i\ge0$ if $(0,0, ..., 0,1)^\T \vec W^{-1} \vec b\le0.$

\end{document}